\documentclass[11pt]{article}

\usepackage{authblk}
\usepackage{fullpage, xcolor}

\usepackage{titlesec}

\usepackage{amssymb}
\usepackage{amsmath}
\usepackage{amsfonts}
\usepackage{amsthm}
\usepackage{autobreak}

\usepackage{enumerate}
\usepackage{float}

\usepackage{indentfirst}

\usepackage{booktabs}
\usepackage{array}
\usepackage{colortbl}
\usepackage{subfigure}

\usepackage{graphicx}

\usepackage[colorlinks,linkcolor=blue,anchorcolor=blue,citecolor=blue]{hyperref}
\usepackage[english]{babel}
\usepackage[T1]{fontenc}
\usepackage{hyperref}

\usepackage{cleveref}
\crefname{ineq}{inequality}{inequalities}
\creflabelformat{ineq}{#2{\upshape(#1)}#3}
\crefname{fact}{fact}{facts}
\crefname{equation}{equation}{equations}
\crefname{algorithm}{algorithm}{algorithm}
\crefname{remark}{remark}{remarks}
\crefname{conjecture}{conjecture}{conjectures}

\usepackage{tikz}
\usetikzlibrary{quantikz}

\usepackage[ruled]{algorithm2e}

\usepackage{enumerate}
\usepackage[title]{appendix}
\usepackage{comment}

\allowdisplaybreaks[2]

\renewcommand{\subparagraph}{\paragraph}

\usepackage{thmtools}

\declaretheorem[style=plain,numberwithin=section]{theorem,lemma,proposition,corollary}
\declaretheorem[style=remark,numberwithin=section]{remark}
\declaretheorem[style=plain,numberwithin=section]{definition,conjecture}

\crefname{ineq}{inequality}{inequalities}
\creflabelformat{ineq}{#2{\upshape(#1)}#3}
\crefname{fact}{fact}{facts}
\crefname{equation}{equation}{equations}
\crefname{algorithm}{algorithm}{algorithm}
\crefname{remark}{remark}{remarks}
\crefname{conjecture}{conjecture}{conjectures}

\usepackage{tikz}
\usetikzlibrary{quantikz}

\newcommand{\PTIME}{\mathsf{P}}

\newcommand{\SZK}{\mathsf{SZK}}
\newcommand{\BPP}{\mathsf{BPP}}
\newcommand{\NP}{\mathsf{NP}}
\newcommand{\coNP}{\mathsf{co\text{-}NP}}
\newcommand{\QMA}{\mathsf{QMA}}

\newcommand{\PP}{\mathsf{PP}}
\newcommand{\SBP}{\mathsf{SBP}}
\newcommand{\MA}{\mathsf{MA}}
\newcommand{\AM}{\mathsf{AM}}
\newcommand{\StoqMA}{\mathsf{StoqMA}}
\newcommand{\coStoqMA}{\mathsf{co\text{-}StoqMA}}
\newcommand{\eStoqMA}{\mathsf{eStoqMA}}
\newcommand{\cStoqMA}{\mathsf{cStoqMA}}
\newcommand{\PSPACE}{\mathsf{PSPACE}}

\renewcommand{\braket}[3]{\langle #1 | #2 | #3 \rangle}
\newcommand{\ketbra}[2]{\left| #1 \right\rangle \left\langle #2 \right|}
\newcommand{\innerprod}[2]{\langle #1 | #2 \rangle}

\newcommand{\norm}[1]{\left\| #1 \right\|}
\newcommand{\Tr}{\mathrm{Tr}}

\newcommand{\poly}{\mathrm{poly}}

\newcommand{\negl}{\mathrm{negl}}
\newcommand{\setunsat}{\mathrm{set\text{-}unsat}}
\newcommand{\SetCSP}{\mathrm{SetCSP}}

\renewcommand{\Pr}[1]{\mathrm{Pr}\left[#1\right]}
\newcommand{\E}{\mathbb{E}}

\newcommand{\binset}{\{0,1\}}
\newcommand{\supp}[1]{\mathrm{supp}\left( #1 \right)}

\newcommand{\calT}{\mathcal{T}}
\newcommand{\calL}{\mathcal{L}}

\newcommand{\argmin}{\mathop{\mathrm{argmin}}}
\newcommand{\Idx}[1]{\mathrm{Idx}\left(#1\right)}

\newcommand{\Pacc}{p_{\rm acc}}

\newcommand{\Sample}[1]{\mathsf{S}_{#1}}
\newcommand{\Query}[1]{\mathsf{Q}_{#1}}

\begin{document}

\title{$\StoqMA$ meets distribution testing}
\author{Yupan Liu\footnote{Email: yupan.liu@gmail.com}}
\date{}
\maketitle

\begin{abstract}
$\StoqMA$ captures the computational hardness of approximating the ground energy of local Hamiltonians that do not suffer the so-called sign problem. 
We provide a novel connection between $\StoqMA$ and distribution testing via reversible circuits. 
First, we prove that easy-witness $\StoqMA$ (viz. $\eStoqMA$, a sub-class of $\StoqMA$) is contained in $\MA$. 
Easy witness is a generalization of a subset state such that the associated set's membership can be efficiently verifiable, and all non-zero coordinates are not necessarily uniform. 
This sub-class $\eStoqMA$ contains $\StoqMA$ with perfect completeness ($\StoqMA_1$), which further signifies a simplified proof for $\StoqMA_1 \subseteq \MA$~\cite{BBT06, BT10}. 
Second, by showing distinguishing reversible circuits with ancillary random bits is $\StoqMA$-complete (as a comparison, distinguishing quantum circuits is $\QMA$-complete~\cite{JWB05}), we construct soundness error reduction of $\StoqMA$. 
Additionally, we show that both variants of $\StoqMA$ that without any ancillary random bit and with perfect soundness are contained in $\NP$. 
Our results make a step towards collapsing the hierarchy $\MA \subseteq \StoqMA \subseteq \SBP$~\cite{BBT06}, in which all classes are contained in $\AM$ and collapse to $\NP$ under derandomization assumptions. 
\end{abstract}

\section{Introduction}

This tale originates from Arthur-Merlin protocols, such as complexity classes $\MA$ and $\AM$, introduced by Babai~\cite{Bab85}. 
$\MA$ is a randomized generalization of the complexity class $\sf NP$, namely the verifier could take advantage of the randomness. 
$\AM$ is additionally allowing two-message interaction. 
Surprisingly, two-message Arthur-Merlin protocols are as powerful as such protocols with a constant-message interaction, whereas it is a long-standing open problem whether $\MA=\AM$. 
It is evident that ${\sf NP} \subseteq {\MA} \subseteq {\AM}$. 
Moreover, under well-believed derandomization assumptions~\cite{KvM02,MV05}, these classes collapse all the way to $\sf NP$. 
Despite limited progresses on proving $\MA=\AM$, is there any intermediate class between $\MA$ and $\AM$? 

$\StoqMA$ is a natural class between $\MA$ and $\AM$, initially introduced by Bravyi, Bessen, Terhal~\cite{BBT06}. 
$\StoqMA$ captures the computational hardness of the stoquastic local Hamiltonian problems. 
The local Hamiltonian problem, defined by Kitaev~\cite{Kit99}, is substantially approximating the minimum eigenvalue (a.k.a. ground energy) of a sparse exponential-size matrix (a.k.a. local Hamiltonian) within inverse-polynomial accuracy. 
Stoquastic Hamiltonians~\cite{BDOT06} are a family of Hamiltonians that do not suffer the sign problem, namely all off-diagonal entries in the Hamiltonian are non-positive. 
$\StoqMA$ also plays a crucial role in the Hamiltonian complexity -- $\StoqMA$-complete is a level in the complexity classification of $2$-local Hamiltonian problems on qubits~\cite{CM16,BH17}, along with $\PTIME$, $\NP$-complete, and $\QMA$-complete. 

Inspiring by the Monte-Carlo simulation in physics, Bravyi and Terhal~\cite{BBT06,BT10} propose a $\MA$ protocol for the stoquastic frustration-free local Hamiltonian problem, which further signifies $\StoqMA$ with perfect completeness ($\StoqMA_1$) is contained in $\MA$. 
A uniformly restricted variant\footnote{It is the projection uniform stoquastic local Hamiltonian problem, namely each local term in Hamiltonian is exactly a projection. See Definition 2.10 in \cite{AGL20}. } of this problem, which is also referred to as $\SetCSP$~\cite{AG20}\footnote{Namely, a modified constraint satisfaction problem such that both constraints and satisfying assignments are a subset. }, essentially captures the $\MA$-hardness. 

To characterize $\StoqMA$ through the distribution testing lens, we begin with an informal definition of $\StoqMA$ and leave the details in \Cref{sec:prelim-classes}. 
For a language $\calL$ in $\StoqMA$, there exists a verifier $V_x$ that takes $x \in \calL$ as an input, where the verifier's computation is given by a classical reversible circuit, viewed as a quantum circuit. 
Besides a non-negative state\footnote{A witness here could be any quantum state, but the optimal witness is a non-negative state, see \Cref{remark:StoqMA-def}. } in the verifier's input as a witness, to utilize the randomness, ancillary qubits in the verifier's input consist of not only state $\ket{0}$ but also $\ket{+} := (\ket{0}+\ket{1})/\sqrt{2}$. 
After applying the circuit, the designated output qubit is measured in the Hadamard basis\footnote{It is worthwhile to mention that we can define $\MA$~\cite{BDOT06} (see \Cref{def:MA}) in the same fashion, namely replacing the measurement on the output qubit by the computational basis. }. 
A problem is in $\StoqMA(a,b)$ for some $a>b \geq 1/2$, if for \textit{yes} instances, there is a witness making the verifier accept with probability at least $a$; whereas for \textit{no} instances, all witness make the verifier accepts with probability at most $b$. 
The gap between $a$ and $b$ is at least an inverse polynomial since error reduction for $\StoqMA$ is unknown. 

The optimality of \textit{non-negative} witnesses suggests a novel connection to \textit{distribution testing}. 
Let $\ket{0}\ket{D_0}+\ket{1}\ket{D_1}$ be the state before the final measurement, where $\ket{D_k}=\sum_{i\in\binset^{n-1}} \sqrt{D_k(i)} \ket{i}$ for $k=0,1$ and $n$ is the number of qubits utilized by the verifier. 
A straightforward calculation indicates that the acceptance probability of a $\StoqMA$ verifier is linearly dependent on the \textit{squared Hellinger distance} $d^2_H(D_0,D_1)$ between $D_0$ and $D_1$, which indeed connects to distribution testing! 
Consequently, to prove $\StoqMA \subseteq \MA$, it suffices to approximate $d^2_H(D_0,D_1)$ within an inverse-polynomial accuracy using merely polynomially many samples\footnote{Each sample is actually the measurement outcome after running an independent copy of the verifier, see \Cref{remark:implement-dual-access}. }.

\subsection{Main results}

\subparagraph{$\StoqMA$ with easy witness ($\eStoqMA$). } 

With this connection to distribution testing, it is essential to take advantage of the \textit{efficient query access} of a non-negative witness where a witness satisfied with this condition is the so-called \textit{easy witness}. 
For this sub-class of $\StoqMA$ (viz. $\eStoqMA$) such that there exists an easy witness for any \textit{yes} instances, we are then able to show an $\MA$ containment by utilizing both query and sample accesses to the witness. 
Informally, easy witness is a generalization of a subset state such that the associated state's membership is efficiently verifiable, and all non-zero coordinates are unnecessarily uniform. 
It is evident that a classical witness is also an easy witness, but the opposite is not necessarily true (See \Cref{remark:eStoqMA-vs-cStoqMA}).
Now let us state our first main theorem: 
\begin{theorem}[Informal of \Cref{thm:eStoqMA-in-MA}]
\label{thm-inf:eStoqMA-eq-MA}
$\eStoqMA = \MA$. 
\end{theorem}
It is worthwhile to mention that easy witness also relates to $\SBP$ (Small Bounded-error Probability)~\cite{BGM06}. 
In particular, Goldwasser and Sipser~\cite{GS86} propose the celebrated Set Lower Bound protocol -- it is an $\AM$ protocol for the problem of approximately counting the cardinality of such an efficient verifiable set. 
Recently, Watson~\cite{Wat16} and Volkovich~\cite{Vol20} separately point out that such a problem is essentially $\SBP$-complete.

Although $\eStoqMA$ seems only a sub-class of $\StoqMA$, we could provide an arguably simplified proof for $\StoqMA_1 \subseteq \MA$ \cite{BBT06}. 
Namely, employed the local verifiability of $\SetCSP$~\cite{AG20}, it is evident to show $\eStoqMA$ contains $\StoqMA$ with perfect completeness, which infers $\StoqMA_1 \subseteq \MA$.
However, it remains open whether all $\StoqMA$ verifier has easy witness, whereas an analogous statement is false for classical witnesses (see \Cref{proposition:StoqMA-classical-witness-not-optimal}). 

\subparagraph{Reversible Circuit Distinguishability is $\StoqMA$-complete. }
It is well-known that distinguishing quantum circuits (a.k.a. the Non-Identity Check problem), namely given two efficient quantum circuits and decide whether there exists a pure state that distinguishes one from the other, is $\QMA$-complete~\cite{JWB05}. 
Moreover, if we restrict these circuits to be reversible (with the same number of ancillary bits), this variant is $\NP$-complete~\cite{Jor14}. 
What happens if we also allow \textit{ancillary random bits}, viewed as quantum circuits with ancillary qubits which is initially state $\ket{+}$? It seems reasonable to believe this variant is $\MA$-complete; however, it is actually $\StoqMA$-complete, as stated in \Cref{thm-inf:distinguishing-circuits-StoqMA-complete}: 
\begin{theorem}[Informal of \Cref{thm:RCD-is-StoqMA-complete}]
	\label{thm-inf:distinguishing-circuits-StoqMA-complete}
	Distinguishing reversible circuits with ancillary random bits within an inverse-polynomial accuracy is $\StoqMA$-complete. 
\end{theorem}

In fact, \Cref{thm-inf:distinguishing-circuits-StoqMA-complete} is a consequence of the distribution testing explanation of a $\StoqMA$ verifier's maximum acceptance probability. 
We can view \Cref{thm-inf:distinguishing-circuits-StoqMA-complete} as new strong evidence of $\StoqMA=\MA$. 
It further straightforwardly inspires a simplified proof of~\cite{Jor14}: 
\begin{proposition}[Informal of \Cref{prop:StoqMA-wo-plus-in-NP}]
	Distinguishing reversible circuits without ancillary random bits is $\NP$-complete.
\end{proposition}

Apart from the role of randomness, \Cref{corr-inf:StoqMA-perfect-soundness} is analogous for $\StoqMA$ regarding the well-known derandomization property~\cite{FGMSZ89} of Arthur-Merlin systems with \textit{perfect soundness}: 
\begin{proposition}[Informal of \Cref{prop:exact-rev-circuit-NP}]
	\label{corr-inf:StoqMA-perfect-soundness}
	$\StoqMA$ with perfect soundness is in $\NP$. 
\end{proposition}

Notably, the $\NP$-containment in \Cref{corr-inf:StoqMA-perfect-soundness} holds even for $\StoqMA(a,b)$ verifiers with arbitrarily small gap $a-b$. 
It is arguably surprising since $\StoqMA(a,b)$ with an exponentially small gap (i.e., the precise variant) at least contains $\NP^{\PP}$ \cite{MN17}, but such a phenomenon does not appear in this scenario. 

\subparagraph{Soundness error reduction of $\StoqMA$. }
Error reduction is a rudimentary property of many complexity classes, such as $\PTIME$, $\BPP$, $\MA$, $\QMA$, etc. . 
It is peculiar that such property of $\StoqMA$ is open, even though this class has been proposed since 2006~\cite{BBT06}. 
An obstacle follows from the limitation of performing a single-qubit Hadamard basis final measurement, so we cannot directly take \textit{the majority vote} of outcomes from the verifier's parallel repetition. 
Utilized the gadget in the proof of \Cref{thm-inf:distinguishing-circuits-StoqMA-complete}, we have derived soundness error reduction of $\StoqMA$, which means we could take \textit{the conjunction} of verifier's parallel repetition's outcomes: 
\begin{theorem}[Soundness error reduction of $\StoqMA$]
\label{thm-inf:soundness-error-reduction-StoqMA}
For any polynomial $r=\poly(n)$, 
\[\StoqMA\left(\frac{1}{2}+\frac{a}{2},\frac{1}{2}+\frac{b}{2}\right) \subseteq \StoqMA\left(\frac{1}{2}+\frac{a^r}{2},\frac{1}{2}+\frac{b^r}{2}\right).\]
\end{theorem}

\subsection{Discussion and open problems}

\subparagraph{Towards $\SBP = \MA$. }
As stated before, it is known $\MA \subseteq \StoqMA \subseteq \SBP \subseteq \AM$~\cite{BGM06,BBT06}. 
Note a subset state associated with an efficient membership-verifiable set is an easy witness. 
Could we utilize this connection and deduce proof of $\SBP \subseteq \eStoqMA$? 

Owing to the wide uses of the Set Lower Bound protocol~\cite{GS86}, such a solution would be a remarkable result with many complexity-theoretic applications. 
Unfortunately, even a $\QMA$ containment for this kind of approximate counting problem is unknown. 
Despite such smart usage of the Grover algorithm implies an $O(\sqrt{2^n/|S|})$-query algorithm~\cite{AR20,BHMT02,VO20}, we are not aware of utilizing a quantum witness. 
Furthermore, an oracle separation between $\SBP$ and $\QMA$~\cite{AKKT19} suggests that such a proof of $\SBP \subseteq \QMA$ is supposed to be in a non-black-box approach, which signifies a better understanding beyond a query oracle is required. 

Besides $\SBP$ vs. $\MA$, it remains open whether $\StoqMA=\MA$. 
It is natural to ask whether each $\StoqMA$ verifier has easy witness. 
However, we even do not know how to prove $\StoqMA(1-a,1-1/\poly(n))$ has easy witness, where $a$ is negligible (i.e., an inverse super-polynomial). 
In~\cite{AGL20}, they prove $\StoqMA(1-a,1-1/\poly(n)) \subseteq \MA$ by applying the probabilistic method on a random walk, whereas the existence of easy witness seems to require a stronger structure\footnote{The candidate here is the set $S$ of all good strings (see \Cref{sec:SetCSP}) of the given $\SetCSP$ instance, which is unnecessary an optimal witness. It is thus unclear whether the frustration of $S$ remains negligible. }. 

\subparagraph{Towards error reduction of $\StoqMA$. }
Error reduction of $\StoqMA$ is an open problem since Bravyi, Bessen, and Terhal define this class in 2006~\cite{BBT06}. 
We first state this conjecture: 
\begin{conjecture}[Error reduction of $\StoqMA$]
	\label{conj:StoqMA-error-reduction}
	For any $a,b$ such that $1/2 \leq b < a \leq 1$ and $a-b \geq 1/\poly(n)$, the following holds for any polynomial $l(n)$:
	$\StoqMA(a,b) \subseteq \StoqMA\left(1-2^{-l(n)},1/2+2^{-l(n)}\right).$
\end{conjecture}

As~\cite{AGL20} shows that $\StoqMA$ with a negligible completeness error is contained in $\MA$, (completeness) error reduction of $\StoqMA$ plays a crucial role in proving $\StoqMA=\MA$. 
Instead of performing the majority vote among parallelly running verifiers, another commonplace approach is first reducing errors of completeness and soundness separately, then utilizing these two procedures alternatively with well-chosen parameters. 
For instance, the renowned polarization lemma of $\SZK$~\cite{SV03,BDRV19}, and the space-efficient error reduction of $\QMA$~\cite{FKLMN16}. 
Since \Cref{thm-inf:soundness-error-reduction-StoqMA} already states soundness error reduction of $\StoqMA$, is it possible to also construct a completeness error reduction? 
Namely, a mechanism that builds a new $\StoqMA(1/2+a'/2,1/2+b'/2)$ verifier from the given $\StoqMA(1/2+a/2,1/2+b/2)$ verifier such that $a'$ is super-polynomially close to $1$. 
It seems to require new ideas since a direct analog of the XOR lemma in the polarization lemma of $\SZK$, such as Lemma 4.11 in~\cite{BDRV19}, does not work here. 

\subparagraph{$\StoqMA$ with exponentially small gap. }
Fefferman and Lin prove~\cite{FL18} that $\sf PreciseQMA$ is as powerful as $\PSPACE$, where $\sf PreciseQMA$ is a variant of $\QMA(a,b)$ with exponentially small gap $a-b$.
Moreover, we know that both $\sf PreciseQCMA$ and $\sf PreciseMA$ are equal to $\NP^{\PP}$~\cite{MN17}, where $\sf PreciseQCMA$ is a precise variant of $\QMA$ with a classical witness of the verifier. 
It is evident that $\sf PreciseStoqMA$ is between $\NP^{\PP}$ and $\PSPACE$, also the classical-witness variant of this class is precisely $\NP^{\PP}$ (see \Cref{subsec:limitations-cStoqMA}). 
Does $\sf PreciseStoqMA$ an intermediate class between $\NP^{\PP}$ and $\PSPACE$, or even strong enough to capture the full $\PSPACE$ power? 

\subsection{Related work}
Guided Stoquastic Local Hamiltonian Problem~\cite{Bravyi15}, which is contained in $\MA$, can be considered a (generalized) Hamiltonian version of $\eStoqMA$.
A guiding state $\ket{\phi}$ of a ground state $\ket{\psi}$ such that $\innerprod{x}{\phi}$\footnote{In fact, Bravyi's $\MA$ containment only requires to efficiently compute $\innerprod{x}{\phi}/\innerprod{y}{\phi}$ for any $x,y \in \binset^n$, which coincides with \Cref{def:eStoqMA}. However, the analysis of this protocol needs to evaluate the amplitude $\innerprod{x}{\psi}$. } is efficiently computable for any $x\in\binset^n$ by a classical circuit of size $p(n)$ and  $\innerprod{x}{\phi} \geq \innerprod{x}{\psi}/p(n)$ where $p(n)$ is a polynomial of $n$. 
This problem connects to $\eStoqMA$ because if a ground state is already a guiding state, then such ground state is evidently easy witness.

\subsection{Paper organization} 
\Cref{sec:prelim} introduces useful terminologies and notations. 
\Cref{sec:eStoqMA} proves that $\eStoqMA$ is contained in $\MA$, which indicates an arguably simplified proof of $\StoqMA_1 \subseteq \MA$, together with remarks on $\cStoqMA$. 
\Cref{sec:reversible-circuit} presents a new $\StoqMA$-complete problem named reversible circuit distinguishability, and the complexity of this problem's exact variant, which infers $\StoqMA$ with perfect soundness is in $\NP$. 
\Cref{sec:AND-repetition-of-StoqMA} provides soundness error reduction for $\StoqMA$.

\section{Preliminaries}
\label{sec:prelim}

\subsection{Non-negative states}
We assume familiarity with quantum computing on the levels of \cite{NC02}. 
Beyond this, we then introduce some notations which are more particular for this paper: 
the \textit{support} of $\ket{\psi}$,  $supp(\ket{\psi}) := \{i \in \{0,1\}^n : \innerprod{\psi}{i} \ne 0\}$, is the set strings with non-zero amplitude.  
A quantum state $\ket{\psi}$ is \textit{non-negative} of $\innerprod{i}{\psi} \geq 0$ for all $i \in \{0,1\}^n$. 
For any $S \subseteq \{0,1\}^n$, we refer to the state $\ket{S} := \frac{1}{\sqrt{|S|}} \sum_{i \in S} \ket{i}$ as the subset state corresponding to the set $S$~\cite{Wat00}. 

\subsection{Complexity class: $\MA$ and $\StoqMA$}
\label{sec:prelim-classes}
A (promise) problem $\calL = (\calL_{\rm yes}, \calL_{\rm no})$ consists of two non-overlapping subsets $\calL_{\rm yes}, \calL_{\rm no} \subseteq \{0,1\}^{*}$. 
These classes $\MA$ and $\StoqMA$ considered in this paper using the language of reversible circuits, as \Cref{def:MA} and \Cref{def:StoqMA}. 
\begin{definition}[$\MA$, adapted from~\cite{BBT06}]
\label{def:MA}
A promise problem $\calL = (\calL_{\rm yes}, \calL_{\rm no}) \in$ $\MA$ if there exists an $\MA$ verifier such that for any input $x \in \calL$, an associated uniformly generated verification circuit $V_x$ using only classical reversible gates (i.e. Toffoli, CNOT, X) on $n:=n_w+n_0+n_+$ qubits and a computational-basis measurement on the output qubit, where $n_w$ is the number of qubits for a witness, and $n_0~(\text{or }n_+)$ is the number of $\ket{0}~(\text{or }\ket{+})$ ancillary qubits, such that 
\begin{description}
    \item[Completeness. ] If $x\in \calL_{\rm yes}$, then there exists an $n$-qubit non-negative witness $\ket{w}$ such that $\Pr{V_x \text{ accepts } \ket{w}} \geq 2/3$. 
    \item[Soundness. ] If $x \in \calL_{\rm no}$, we have $\Pr{V_x \text{ accepts } \ket{w}} \leq 1/3$ for any $n$-qubit witness $\ket{w}$. 
\end{description}
\end{definition}

For simplicity, we denote $\ket{\bar{0}}:=\ket{0}^{\otimes n_0}$ and $\ket{\bar{+}}:=\ket{+}^{\otimes n_+}$ for the rest of this paper. 
We refer the equivalence between \Cref{def:MA} and the standard definition of $\MA$ to as \Cref{remark:MA-equivalence}, which is first observed by~\cite{BDOT06}. 
\begin{remark}[Equivalent definitions of $\MA$]
\label{remark:MA-equivalence}
The standard definition of $\MA$ only allows classical witnesses, viz. binary strings. 
To show it is equivalent to \Cref{def:MA}, it suffices to prove the optimal witness for \textit{yes} instances is classical. 
Notice that $\Pr{V_x \text{ accepts } \ket{w}} = \bra{\psi_{\rm in}} V_x^{\dagger} \Pi_{\rm out} V_x \ket{\psi_{\rm in}}$ where $\ket{\psi_{\rm in}}:=\ket{w}\otimes\ket{\bar{0}}\otimes\ket{\bar{+}}$ and $\Pi_{\rm out} = \ketbra{0}{0}_1\otimes I_{\rm else}$. 
Since $V_x^{\dagger} \Pi_{\rm out} V_x$ is a diagonal matrix, the optimal witness of $V_x$ is classical. 
\end{remark}

Analogously, we could define $\NP$ using classical reversible gates by setting $n_+=0$ in \Cref{def:MA}. 
Now we proceed with the definition of $\StoqMA$. 

\begin{definition}[$\StoqMA$, adapted from~\cite{BBT06}]
\label{def:StoqMA}
A promise problem $\calL = (\calL_{\rm yes}, \calL_{\rm no}) \in$ $\StoqMA$ if there is a $\StoqMA$ verifier such that for any input $x \in \calL$, a uniformly generated verification circuit $V_x$ using Toffoli, CNOT, X gates on $n:=n_w+n_0+n_+$ qubits and a Hadamard-basis measurement on the output qubit, where $n_w$ is the number of qubits for a witness, and $n_0$ $(\text{or}~n_+)$ is the number of $\ket{0}$ $(\text{or}~\ket{+})$ ancillary qubits, such that for efficiently computable functions $a(n)$ and $b(n)$: 
\begin{description}
    \item[Completeness. ] If $x\in \calL_{\rm yes}$, then there exists an $n$-qubit non-negative witness $\ket{w}$ such that $\Pr{V_x \text{ accepts } \ket{w}} \geq a(n)$. 
    \item[Soundness. ] If $x \in \calL_{\rm no}$, we have $\Pr{V_x \text{ accepts } \ket{w}} \leq b(n)$ for any $n$-qubit witness $\ket{w}$. 
\end{description}
Moreover, $a(n)$ and $b(n)$ satisfy $1/2 \leq b(n) < a(n) \leq 1$ and $a(n)-b(n) \geq 1/\poly(n)$. 
\end{definition}
Error reduction of $\StoqMA$ remains open since this class was defined in 2006 \cite{BBT06} because this class does not permit amplification of gap between thresholds $a, b$ based on majority voting. 
Hence, this gap is at least an inverse polynomial. 
We leave the remarks regarding the non-negativity of witnesses and parameters to \Cref{remark:StoqMA-def}. 
\begin{remark}[Optimal witnesses of a $\StoqMA$ verifier is non-negative]
\label{remark:StoqMA-def}
Analogous to $\QMA$, the maximum acceptance probability of a $\StoqMA$ verifier $V_x$ is precisely the maximum eigenvalue of $M_x:=\bra{\bar{0}}\bra{\bar{+}}V_x^{\dagger} \ketbra{+}{+}_1 V_x \ket{\bar{0}}\ket{\bar{+}}$ due to $\Pr{V_x \text{ accepts } \ket{\psi}} = \bra{\psi} M_x \ket{\psi}$. 
Notice the matrix $M_x$ is entry-wise non-negative. 
Owing to the Perron-Frobenius theorem (see Theorem 8.4.4 in \cite{HJ12}), a straightforward corollary is that the eigenvector $\psi$ (i.e., the optimal witness) maximizing the acceptance probability has non-negative amplitudes in the computational basis, namely it suffices to consider only non-negative witness for \textit{yes} instances. 
Additionally, it is clear-cut that the acceptance probability for any non-negative witness $\ket{\psi}$, regardless of the optimality, is at least $1/2$ by a direct calculation. 
\end{remark}

\subsection{Distribution testing}
\label{sec:distribution-testing}
Distribution testing is generally about telling whether one probability distribution is close to the other. 
We further recommend a comprehensive survey \cite{Can20} for a detailed introduction. 
We begin with the squared Hellinger distance $d^2_H(D_0,D_1)$ between two (sub-)distributions $D_0,D_1$, where $d^2_H(D_0,D_1):=\frac{1}{2} \| \ket{D_0}-\ket{D_1}\|_2^2$ and $\ket{D_k} = \sum_{i} \sqrt{D_k(i)}\ket{i}$ for any $k=0,1$. 
This distance is comparable with the total variation distance (see Proposition 1 in~\cite{DKW18}). 
We then introduce a specific model used for this paper, namely the \textit{dual access model}: 

\begin{definition}[Dual access model, adapted from~\cite{CR14}]
	\label{def:dual-access-model}
	Let $D$ be a fixed distribution over $[2^n]$. A dual oracle for $D$ is a pair of oracles $(\Sample{D},\Query{D})$: 
	\begin{itemize}
		\item{Sample access:} $\Sample{D}$ returns an element $i\in\binset^{n}$ with probability $D(i)$. And it is independent of all previous calls to any oracle. 
		\item{Query access:} $\Query{D}$ takes an input a query element $j\in \binset^{n-1}$, and returns the quotient $D(0||j)/D(1||j)$ where $D(a||j)$ is the probability weight that $D$ puts on $a||j$ for $a \in \binset$. 		
	\end{itemize}
\end{definition}

We then explain how to implement these oracles here in \Cref{remark:implement-dual-access}: 
\begin{remark}[Implementation of dual access model]
\label{remark:implement-dual-access}
The sample access oracle in \Cref{def:dual-access-model} could be implemented by running an independent copy of the circuit that generates the state $\ket{0}\ket{D_0}+\ket{1}\ket{D_1}$, and measuring all qubits on the computational basis. Meanwhile, the query access oracle is substantially an efficient evaluation algorithm corresponding to the quotient $D_0(i)/D_1(i)$ for given index $i$. 
\end{remark}

In \cite{CR14}, Canonne and Rubinfeld show that approximating the total variation distance between two distributions within an additive error $\epsilon$ requires only $\Theta(1/\epsilon^2)$ oracle accesses (see Theorems 6 and 7 in \cite{CR14}). 
However, suppose we allow to utilize only sample accesses. In that case, such a task requires $\Omega(N/\log N)$ samples even within constant accuracy (see Theorem 9 in~\cite{DKW18}), where $N$ is the dimension of distributions. 

\section{$\StoqMA$ with easy witnesses}
\label{sec:eStoqMA}

This section will prove that $\StoqMA$ with easy witnesses, viz. $\eStoqMA$, is contained in $\MA$. 
\textit{Easy witness} is named in the flavor of the seminal \textit{easy witness lemma}~\cite{IKW02}, which means that an $n$-qubit non-negative state witness of a $\StoqMA$ verifier has a \textit{succinct} representation. 
In particular, there exists an efficient algorithm to output the quotient $D_0(i)/D_1(i)$ for given index $i$. 
It is a straightforward generalization of subset states where the membership of the corresponding subset is efficiently verifiable. 
We here define $\eStoqMA$ formally: 

\begin{definition}[$\eStoqMA$]
	\label{def:eStoqMA}
	A promise problem $\calL=(\calL_{\rm yes},\calL_{\rm no}) \in \eStoqMA$ if there is a $\StoqMA$ verifier such that for any input $x \in \calL$, a uniformly generated verification circuit $V_x$ using only Toffoli, CNOT, X gates on $n:=n_w+n_0+n_+$ qubits and a Hadamard-basis measurement on the output qubit, where $n_w$ is the number of qubits for a witness, and $n_0~(\text{or } n_+)$ is the number of $\ket{0}~(\text{or } \ket{+})$ ancillary qubits, such that for efficiently computable functions $a(n)$ and $b(n)$: 
	\begin{description}
		\item[Completeness.] There exists an $n$-qubit non-negative witness $\ket{w}:=\sum_{i\in\binset^n} \sqrt{D_w(i)} \ket{i}$ such that $\Pr{V_x \text{ accepts } \ket{w}} \geq a(n)$, and there is an efficient algorithm $\Query{w}$ that outputs $D_w(0||i)/D_w(1||i)$ (or $D_w(1||i)/D_w(0||i)$) of index $1||i$ (or $0||i$) sampled from the distribution $D_w$ where $i\in\binset^{n-1}$. 
		\item[Soundness.] For any $n$-qubit witness $\ket{w}$, $\Pr{V_x \text{ accepts }\ket{w}} \leq b(n)$. 
	\end{description}
	Moreover, $a(n)$ and $b(n)$ satisfy $1/2 \leq b(n) < \alpha(n) \leq 1$ and $a(n) - b(n) \geq 1/\poly(n)$. 
\end{definition}

\begin{remark}[Subset-state witnesses require only membership]
	\label{remark:subset-state-only-membership}
	To show a subset-state witness $\ket{w}$ is an easy witness, it suffices to decide the membership of $\supp{\ket{w}}$ for the associated algorithm $\Query{w}$. 
	Notice any coordinate $D_w(j)$ in $D_w$ is $1/|\supp{\ket{w}}|$ if $j \in \supp{\ket{w}}$; otherwise $D_w(j)=0$. 
	Moreover, if $D_w(1||i)=0$ for some $i$, the corresponding point will never be sampled. 
	Hence, the quotient $D_w(0||i)/D_w(1||i)$ is $1$ if both $0||i$ and $1||i$ belong to $\supp{\ket{w}}$ (i.e., $D_w(0||i) = D_w(1||i) \neq 0$); otherwise the quotient is $0$. 	
\end{remark}

Distribution testing techniques inspire an $\MA$ containment of $\eStoqMA$, as \Cref{thm:eStoqMA-in-MA}. 
Precisely, employed with the dual access model (see \Cref{def:dual-access-model}) adapted from Canonne and Rubinfeld~\cite{CR14}, we obtain an empirical estimation within inverse-polynomial accuracy of an $\eStoqMA$ verifier's acceptance probability, where both sample complexity and time complexity are efficient. 

\begin{theorem}[$\eStoqMA\subseteq\MA$]
	\label{thm:eStoqMA-in-MA}
	For any $1/2 \leq b < a \leq 1$ and $a-b\geq 1/\poly(n)$, 
	\[\eStoqMA(a,b) \subseteq \MA\left(\tfrac{9}{16},\tfrac{7}{16}\right).\]
\end{theorem}

In~\cite{BBT06,BT10}, Bravyi, Bessen, and Terhal proved $\StoqMA_1 \subseteq \MA$, utilizing a relatively complicated random walk based argument. 
By taking advantage of $\eStoqMA$, we here provide an arguably simplified proof by plugging \Cref{prop:StoqMA1-in-eStoqMA1} into \Cref{thm:eStoqMA-in-MA}: 
\begin{proposition}
    \label{prop:StoqMA1-in-eStoqMA1}
	$\StoqMA_1 \subseteq \eStoqMA$. 
\end{proposition}
The proof of \Cref{prop:StoqMA1-in-eStoqMA1} straightforwardly follows from the definition of $\SetCSP$ (see \Cref{def:SetCSP}), namely any $\SetCSP_{0,1/\poly}$ instance certainly has easy witness, and it is indeed optimal. 
We further leave the technical details regarding $\SetCSP$ in \Cref{sec:SetCSP}. 

\vspace{1em}
How strong is the $\eStoqMA$? 
\Cref{remark:eStoqMA-vs-cStoqMA} suggests $\eStoqMA$ seems more powerful than classical-witness $\StoqMA$ (i.e., $\cStoqMA$): 
\begin{remark}[$\eStoqMA$ is not trivially contained in $\cStoqMA$]
\label{remark:eStoqMA-vs-cStoqMA}
Classical witness is clearly also easy witness, but the opposite is unnecessarily true. 
Even though Merlin could send the algorithm $\Query{D_w}$ as classical witness to Arthur, Arthur only can prepare $\ket{w}$ by a post-selection, which means $\cStoqMA$ does not trivially contain $\eStoqMA$. 
\end{remark}

Furthermore, the proof of $\StoqMA(a,b)$ with classical witnesses is in $\MA$~\cite{Grilo20} could preserve completeness and soundness parameters. By inspection, it is clear-cut that this proof even holds when the gap $a-b$ is \textit{arbitrarily small}, whereas the proof of \Cref{thm:eStoqMA-in-MA} works only for inverse-polynomial accuracy. 
Further remarks of classical witness' limitations can be found in \Cref{subsec:limitations-cStoqMA}. 

\subsection{$\eStoqMA\subseteq \MA$: the power of distribution testing}
To derive an $\MA$ containment of $\eStoqMA$, it suffices to distinguish two non-negative states (viz., approximating the maximum acceptance probability) within an inverse-polynomial accuracy regarding the inner product (i.e., squared Hellinger distance). 
It seems plausible to prove $\StoqMA\subseteq\MA$ by taking \textit{samples} and post-processing. 
However, the known sample complexity lower bound (See \Cref{sec:distribution-testing}) indicates that (almost) exponentially many samples are unavoidable. 
Fortunately, we could circumvent this barrier for showing $\eStoqMA \subseteq \MA$, since easy witness guarantees \textit{efficient query access} to $D_0(i)/D_1(i)$ for given index $i$. 
In particular, employing both sample and query oracle accesses to $D_0,D_1$, such approximation within an additive error $\epsilon$ requires merely $\Theta(1/\epsilon^2)$ samples and queries! 
This advantage first noticed by Rubinfeld and Servedio~\cite{RS09}, and then almost fully characterized by Canonne and Rubinfeld~\cite{CR14}.  
Recently, this technique also has algorithmic applications used in quantum-inspired classical algorithms for machine learning~\cite{CGLLTW19, Tang19}. 

\begin{lemma}[Approximating a single-qubit Hadamard-basis measurement]
\label{lemma:simulating-X-basis-measurement}
In the dual access model, there is a randomized algorithm $\calT$ which takes an input $x$, $1/2 \leq b(|x|) < a(|x|) \leq 1$, as well as access to $(\Sample{D},\Query{D})$, where the non-negative state before the measurement is $\ket{\psi} = \sum_{i\in[2^n]} \sqrt{D(i)} \ket{i}$. After making $O\left(1/(a-b)^2\right)$ calls to the oracles, $\calT$ outputs either \textsc{accept} or \textsc{reject} such that: 
\begin{itemize}
	\item If $\frac{1}{2}\|\ket{D_0}+\ket{D_1}\|_2^2 \geq a$, $\calT$ outputs \textsc{accept} with probability at least $9/16$; 
	\item If $\frac{1}{2}\|\ket{D_0}+\ket{D_1}\|_2^2 \leq b$, $\calT$ outputs \textsc{accept} with probability at most $7/16$, 
\end{itemize}
where $D_k~(k\in\binset)$ is a sub-distribution such that $\forall i \in \binset^{n-1}, D_k(i):=D(k||i)$. 
\end{lemma}

\subparagraph{Proof Intuition.} 
	To construct this algorithm $\calT$, the main idea is writing the acceptance probability $\Pacc$ of a $\StoqMA$ verifier's easy witness as an expectation over $D_1$ (or $D_0$) of some random variable regarding coordinates quotients $D_0(i)/D_1(i)$. 
	Note that the quotient $\sqrt{D_0(i)}/\sqrt{D_1(i)}$ could be computed by running the evaluation algorithm $Q_w$ (i.e., query oracle access). 
	Hence, $\calT$ only require to calculate an empirical estimation of $\E[X]$ (see the RHS of \Cref{eq:statistics-approx-PaccX}) within $1/\poly(|x|)$ accuracy. 
	Such an approximation could be achieved by averaging $\poly(|x|)$ sample with a standard concentration bound, which is analogous to Theorem 6 in~\cite{CR14}. 
	
Now we proceed with the explicit construction (i.e., \Cref{algo:PaccX-tester}) and analysis. 
	
\begin{proof}[Proof of \Cref{lemma:simulating-X-basis-measurement}. ]
	We begin with estimating the quantity $\norm{\ket{D_0}+\ket{D_1}}_2^2/2\norm{D_1}_1$ up to some additive error $\epsilon := (a-b)/8$. 
	We first observe that
	\begin{equation}
	\label{eq:statistics-approx-PaccX}
	\tfrac{\norm{\ket{D_0}+\ket{D_1}}_2^2}{2\norm{D_1}_1} 
	= \frac{1}{2} \sum_{i \in \binset^{n-1}} \left(1+\tfrac{\sqrt{D_0(i)}}{\sqrt{D_1(i)}}\right)^2 \tfrac{D_1(i)}{\|D_1\|_1}\\
	= \mathop{\mathbb{E}}_{i \sim D_1/\norm{D_1}_1} \left[ \frac{1}{2} \left( 1+\tfrac{\sqrt{D_0(i)}}{\sqrt{D_1(i)}} \right)^2 \right]. 
	\end{equation}
	Since the inner product is symmetric, it also implies 
	$\frac{\norm{\ket{D_0}+\ket{D_1}}_2^2}{2\norm{D_0}_1} = \mathop{\mathbb{E}}_{i \sim \tfrac{D_0}{\norm{D_0}_1}} \left[ \frac{1}{2} \left( 1+\frac{\sqrt{D_1(i)}}{\sqrt{D_0(i)}} \right) \right]$. 
	
	Notice $\calT$ only require to achieve an empirical estimate of this expected value, which suffices to utilize $m=O\left(1/(a-b)^2\right)$ samples $s_i$ from $D_1$, querying $\tfrac{D_0(s_i)}{D_1(s_i)}$, and computing $X_i = \tfrac{1}{2}\left( 1+\tfrac{\sqrt{D_0(s_i)}}{\sqrt{D_1(s_i)}} \right)^2 \|D_1\|_1$. 
	We here provide the explicit construction of $\calT$, as \Cref{algo:PaccX-tester}. 
	
	\begin{algorithm}[t!]
		\SetKwInOut{Input}{Require} 
		\SetKwFor{For}{For}{Do}{End}
		\SetKwIF{If}{ElseIf}{Else}{If}{Then}{Else If}{Else}{endif}
		\SetAlgoLined
		\Input{$\Sample{D}$ and $\Query{D}$ oracle accesses; parameters $\frac{1}{2} \leq b < a \leq 1$. }
		Set $m,m' := \Theta(1/\epsilon^2)$, where $\epsilon := (a-b)/8$\;
		Draw samples $o_1,\cdots,o_{m'}$ from $D_{\rm out}:= \text{marginal distribution of the designated output qubit}$\;
		Compute $\hat{Z} := \frac{1}{m'} \sum_{i=1}^{m'} Z_i$, where $Z_i := o_{i}$\;
		Draw samples $s_1,\cdots,s_{m}$ from $D$\;
		\For{$i = 1,\cdots,m$}{
			\lIf{$\hat{Z} \geq \frac{1}{2}$}{with $\Query{D}$, get $X_i := \frac{1}{2} \left(1+\tfrac{\sqrt{D_0(s_i)}}{\sqrt{D_1(s_i)}}\right)^2$}
			\lElse*{with $\Query{D}$, get $X_i := \frac{1}{2} \left(1+\tfrac{\sqrt{D_1(s_i)}}{\sqrt{D_0(s_i)}}\right)^2$}\;
		}
		Compute $\hat{X} := \frac{1}{m} \sum_{i=1}^m X_i$\;
		\lIf{$\hat{Z} \geq \frac{1}{2}$ \text{ and}  $\hat{X}\hat{Z} \geq \frac{1}{2}(a+b)$}{output ACCEPT}
		\lElseIf{$\hat{Z} < \frac{1}{2}$ \text{ and } $\hat{X}(1-\hat{Z}) \geq \frac{1}{2}(a+b)$}{output ACCEPT}
		\lElse{output REJECT}
		\caption{$O(1/(a-b)^2)$-additive approximation tester $\calT$ of $\frac{1}{2}\norm{\ket{D_0}+\ket{D_1}}_2^2$}
		\label[algorithm]{algo:PaccX-tester}
	\end{algorithm}
	
	\subparagraph{Analysis. } Define random variables $Z_i$ as in \Cref{algo:PaccX-tester}. 
	We obviously have $\E[Z_i] = \|D_1\|_1 \in [0,1]$. Since all $Z_i$s' are independent, a Chernoff bound ensures 
	\begin{equation}
		\label{eq:eStoqMA-in-MA-RV1}
		\Pr{\left|\hat{Z} - \|D_1\|_1\right| \leq \epsilon} \geq 1-2e^{-2m'/\epsilon^2},
	\end{equation}
	which is at least $3/4$ by an appropriate choice of $m'$. 
	
	Note drawing samples from $p_0$ implicitly by post-selecting the output qubit to be $0$. 
	However, due to the inner product's symmetry and $\norm{D_0}_1+\norm{D_1}_1=1$, there must exist $i \in \binset$ such that $\norm{D_i}_1 \geq 1/2$. Hence, the required sample complexity will be enlarged merely by a factor of $2$. 
	
	Let us also define random variables $X_i$ as in \Cref{algo:PaccX-tester}. W.L.O.G. assume that $\norm{D_1}_1 \geq 1/2 \geq \norm{D_0}_1$. 
	By \Cref{eq:statistics-approx-PaccX}, we obtain
	$\E_{i \sim D_1/\norm{D_1}_1} [X_i] = \norm{\ket{D_0}+\ket{D_1}}_2^2/2\norm{D_1}_1$. 
	Because the $X_i$'s are independent and takes value in $[1/2,1]$, by Chernoff bound,
	\begin{equation}
		\label{eq:eStoqMA-in-MA-RV2}
		\Pr{\left|\hat{X}-\frac{\norm{\ket{D_0}+\ket{D_1}}_2^2}{2\|D_1\|_1}\right| \leq \epsilon} \geq 1- 2e^{-2m/\epsilon^2}.
	\end{equation}
	Therefore, by our choice of $m$, $\hat{X}$ is an $\epsilon$-additive approximation of $\norm{\ket{D_0}+\ket{D_1}}_2^2/2\norm{D_1}_1$ with probability at least $3/4$. 
	Note that $X_i, Z_i$ are independent, we obtain $\E\left[\hat{X} \hat{Z}\right] = \frac{1}{2} \norm{\ket{D_0}+\ket{D_1}}_2^2$. 
	Hence, notice $1/2 \leq \|D_1\|_1 \leq 1$ and $1/2 \leq \frac{1}{2}\norm{\ket{D_0}+\ket{D_1}}_2^2 \leq 1$, by combining Equations \eqref{eq:eStoqMA-in-MA-RV1} and \eqref{eq:eStoqMA-in-MA-RV2}, we obtain with probability $9/16$:
	\[
	\begin{aligned}
		\hat{X}\hat{Z} &\leq \left( \tfrac{\norm{\ket{D_0}+\ket{D_1}}_2^2}{2\|D_1\|_1} + \epsilon \right)\left( \|D_1\|_1+\epsilon \right)
		\leq \tfrac{1}{2}\norm{\ket{D_0}+\ket{D_1}}_2^2 +\epsilon^2+\epsilon+2\epsilon \leq \tfrac{1}{2}\norm{\ket{D_0}+\ket{D_1}}_2^2+4\epsilon; \\
		\hat{X}\hat{Z} &\geq \left( \tfrac{\norm{\ket{D_0}+\ket{D_1}}_2^2}{2\|D_1\|_1} - \epsilon \right)\left( \|D_1\|_1-\epsilon \right)
		\geq \tfrac{1}{2}\norm{\ket{D_0}+\ket{D_1}}_2^2+\epsilon^2-\epsilon-2\epsilon \geq \tfrac{1}{2}\norm{\ket{D_0}+\ket{D_1}}_2^2-4\epsilon. 
	\end{aligned}
	\]
	It implies that
	$\Pr{\left|\hat{X}\hat{Z} - \tfrac{1}{2}\norm{\ket{D_0}+\ket{D_1}}_2^2\right| \leq 4\epsilon} \geq 9/16$. 
	We thereby conclude that 
	\begin{itemize}
		\item If $\tfrac{1}{2}\norm{\ket{D_0}+\ket{D_1}}_2^2 \geq a$, then $\hat{X}\hat{Z} \geq a-4\epsilon$ and $\calT$ outputs ACCEPT w.p. at least $9/16$. 
		\item If $\tfrac{1}{2}\norm{\ket{D_0}+\ket{D_1}}_2^2 \leq b$, then $\hat{X}\hat{Z} \leq b+4\epsilon$ and $\calT$ outputs ACCEPT w.p. at most $7/16$. 
	\end{itemize}
	
	Furthermore, the algorithm $\calT$ makes $m'+2m$ calls for $\Sample{D}$ and $m$ calls for $\Query{D}$ . 
\end{proof} 

It is worthwhile to mention that this construction in the proof of \Cref{thm:eStoqMA-in-MA} is optimal regarding the sample complexity, as Theorem 7 stated in~\cite{CR14}. 

\vspace{1em}
Finally, we complete the proof of \Cref{thm:eStoqMA-in-MA} by \Cref{lemma:simulating-X-basis-measurement}. 

\begin{proof}[Proof of \Cref{thm:eStoqMA-in-MA}. ]
	Given an $\eStoqMA(a,b)$ verifier $V_x$, we here construct a $\MA$ verifier $V'_x$ that follows from \Cref{algo:PaccX-tester} in the proof of \Cref{lemma:simulating-X-basis-measurement}: 
	
\begin{enumerate}[(1)]
	\item For each call to the sample oracle $\Sample{D_w}$, we run the $\eStoqMA$ verifier $V_x$ (without measuring the output qubit) with the witness $w$, and draw samples by performing measurements: 
	\begin{itemize}
		\item For samples $s_i~(1 \leq i \leq m)$ from distribution $D$, measure all qubits utilized by the verification circuit in the computational basis; 
		\item For samples $o_j~(1 \leq j \leq m')$ from distribution $D_{\rm out}$, measure the designated output qubit in the computational basis.  
	\end{itemize}
	\item For each call to the query oracle $\Query{D_w}$ with index $i$, find the corresponding index $i'$ at the beginning by performing the permutation associated with $V_x^{\dagger}$ on $i$, and then evaluate the value $D_w(i'')/D_w(i')$ by utilizing the given algorithm associated with this easy witness, where $i''$ is given by flipping the first bit of $i'$. 
	\item Compute an empirical estimation of $\frac{1}{2} \norm{\ket{D_0}+\ket{D_1}}_2^2$ as \Cref{algo:PaccX-tester}, and then decide whether $V_x$ accepts $w$.   
\end{enumerate}

The circuit size of $V'_x$ is a polynomial of $|x|$ since both sample and query complexity are efficient. 
We thus conclude that the new $\MA$ verifier $V'_x$ is efficient, and only requires $O\left(1/(a-b)^2\right)$ copies of the witness $w$, which finishes the completeness case. 

For the soundness case, the acceptance probability $\Pacc$ of the $\eStoqMA$ verifier $V_x$ for all witnesses is obviously upper-bounded by $b$, regardless of whether such a witness is easy or not. Furthermore, entangled witnesses are useless since we draw samples by performing measurements separately. Hence, the maximum acceptance probability of the new $\MA$ verifier $V'_x$ is also at most $b$. 
\end{proof}

\subsection{$\StoqMA$ with perfect completeness is in $\eStoqMA$}
\label{subsec:StoqMA1-in-eStoqMA}

We here complete proof of \Cref{prop:StoqMA1-in-eStoqMA1}. By \Cref{thm:eStoqMA-in-MA}, it infers $\StoqMA_1 \subseteq \MA$. 

\begin{proof}[Proof of \Cref{prop:StoqMA1-in-eStoqMA1}. ]
	By \Cref{SetCSP-wo-frustration-is-StoqMA1-complete}, we know that $\SetCSP_{0,1/\poly}$ is $\StoqMA_1$-complete, so it suffices to show that $\SetCSP_{0,1/\poly}$ is contained in $\eStoqMA_1$.

	By \Cref{lemma:SetCSP-in-StoqMA}, given a $\SetCSP_{0,b}$ instance $C$, we can construct a $\StoqMA\left(1,1-b/2\right)$ verifier. 
	The corresponding subset $S \subseteq \binset^n$, where $S$ satisfies all set-constraints of $C$, is an optimal witness. 
	It is left to show that this subset states is an easy witness. 
	
	We achieve the proof by inspection. 
	Let $S$ be the set of all good strings of $C$, then $\setunsat(C,S)=0$. Note $x\in S$ is a good string of $C$ iff $x$ is a good string of all set-constraints $C_i (1\leq i \leq m)$, the membership of $S$ thus can be decided efficiently, which infers the subset state $\ket{S}$ is easy witness by \Cref{remark:subset-state-only-membership}. 
\end{proof}

\subsection{Limitations of classical-witness $\StoqMA$}
\label{subsec:limitations-cStoqMA}

As we have shown $\StoqMA$ with easy witness is contained in $\MA$. What about classical witness, namely $\cStoqMA$? In fact, we could show such a containment that preserves both completeness and soundness parameters. 
\begin{proposition}[\cite{Grilo20}]
	\label{prop:cStoqMA-in-MA}
	For any $1/2 \leq b < a \leq 1$ and $a-b\geq 1/\poly(n)$, $\cStoqMA(a,b) \subseteq \MA(2a-1,2b-1).$
\end{proposition}

\begin{proof}[Proof Sketch. ]
We only illustrate the intuition: for any $s\in\binset^n$ and any reversible circuit $U$, we have $\bra{s} U^{\dagger} \ketbra{+}{+}_1 U \ket{s} = \tfrac{1}{2} + \tfrac{1}{2} \bra{s} U^{\dagger} X_1 U \ket{s}$ since $\ketbra{+}{+}=\frac{1}{2} (X+I)$. 
The detailed proof is left in \Cref{subsec:cStoqMA-proofs}. 
\end{proof}

The proof of \Cref{prop:cStoqMA-in-MA} immediately infers the \textit{precise variant} of $\StoqMA$ with classical witnesses, where the completeness-soundness gap is exponentially small, is equal to $\mathsf{PreciseMA}$. 
However, the proof of \Cref{thm:eStoqMA-in-MA} no longer works for precise scenarios, indicating that $\StoqMA$ with classical witness seems not interesting. 

Furthermore, it is not hard to see that classical witness is optimal for $\StoqMA_1$ verifier\footnote{By combining $\StoqMA_1\subseteq \MA_1$ and the gadget in the proof of \Cref{proposition:simulating-Z-meas-by-X}, we could construct a $\StoqMA_1$ verifier such that a classical witness is optimal. }.
However, it does not mean that a classical witness is optimal for \textit{any} $\StoqMA_1$ verifier. 
In fact \Cref{subsec:StoqMA-classical-witness-not-optimal} provides a simple counterexample by considering an identity as a verifier. 
However, this impossibility result is unknown for easy witness yet.

\section{Complexity of reversible circuit distinguishability}
\label{sec:reversible-circuit}

This section will concentrate on the complexity classification of distinguishing reversible circuits, namely given two efficient reversible circuits, and decide whether there is a non-negative state that \textit{cannot} tell one from the other.
With ancillary random bits, this problem is $\StoqMA$-complete, as \Cref{thm:RCD-is-StoqMA-complete}.
However, this problem's exact variant, namely assuming two reversible circuits are indistinguishable with respect to any non-negative witness for \textit{no} instances (viz., $\StoqMA$ with perfect soundness), is $\NP$-complete (see \Cref{prop:exact-rev-circuit-NP}).
Moreover, \Cref{thm:RCD-is-StoqMA-complete} also implies that distinguishing reversible circuits without any ancillary random bit is $\NP$-complete, which signifies a simplified proof of~\cite{Jor14}. 

\subsection{Reversible circuit distinguishability is $\StoqMA$-complete}
\label{subsec:reversible-circuit-StoqMA}

We begin with the formal definition of the \textit{Reversible Circuit Distinguishability} problem. 

\begin{definition}[Reversible Circuit Distinguishability]
\label{def:non-closeness-check}
Given a classical description of two reversible circuits $C_0,C_1$ (using Toffoli, CNOT, X gates) on $n:=n_w+n_0+n_+$ qubits, where $n_w$ is the number of qubits of a non-negative state witness $\ket{w}$, $n_0$ is the number of $\ket{0}$ ancillary qubits, and $n_+$ is the number of $\ket{+}$ ancillary qubits. Let the resulting state before measuring the  output qubit be $\ket{R_i}:=C_i \ket{w}\ket{\bar{0}}\ket{\bar{+}}$, $i\in\binset$. 
Promise that $C_0$ and $C_1$ with respect to witness state(s) are either $\alpha$-indistinguishable or $\beta$-distinguishable, decide whether 
\begin{itemize}
	\item {\bf Yes} ($\alpha$-indistinguishable): there exists a non-negative witness $\ket{w}$ such that $\innerprod{R_0}{R_1} \geq \alpha$;
	\item {\bf No} ($\beta$-distinguishable): for any non-negative witness $\ket{w}$,  then $\innerprod{R_0}{R_1} \leq \beta$;
\end{itemize}
where $\alpha-\beta \geq 1/\poly(n)$\footnote{Note $\innerprod{R_0}{R_0}=\innerprod{R_1}{R_1}=1$ which differs from $\innerprod{D_0}{D_0}+\innerprod{D_1}{D_1}=1$ previously used in \Cref{sec:eStoqMA}, we obtain that the acceptance probability $p_{\rm acc}=\frac{1}{2}+\frac{1}{2}\innerprod{R_0}{R_1}=1-\frac{1}{2}\cdot \frac{1}{2}\|\ket{R_0}-\ket{R_1}\|_2^2$. }.  
\end{definition}

Since \Cref{def:non-closeness-check} seems slightly inconsistent with known results regarding distinguishing circuits \cite{JWB05,Jor14,Tan10}, it is worthwhile to mention a slightly different version (see \Cref{remark:RCD-coStoqMA}) of \Cref{def:non-closeness-check}, which is $\coStoqMA$-complete. 
\begin{remark}[Equivalence Check of Reversible Circuits is $\coStoqMA$-complete]
	\label{remark:RCD-coStoqMA}
	Consider the same scenario in \Cref{def:non-closeness-check}, and the task is checking whether $C_0$ and $C_1$ are approximately equivalent (with respect to witness states). More concretely, decide whether $\innerprod{R_0}{R_1} \geq \alpha$ for any $\ket{w}$; or there exists $\ket{w}$ such that $\innerprod{R_0}{R_1} \leq \beta$. The $\coStoqMA$-completeness straightforwardly follows from the constructions in the proof of \Cref{thm:RCD-is-StoqMA-complete}.  
\end{remark}

Now we state the main theorem in \Cref{sec:reversible-circuit}. 
\begin{theorem}[Reversible Circuit Distinguishability is $\StoqMA$-complete]
\label{thm:RCD-is-StoqMA-complete}
	For any $\alpha - \beta \geq 1/\poly(n)$, $(\alpha,\beta)$-Reversible Circuit Distinguishability is $\StoqMA\left(1/2+\alpha/2,1/2+\beta/2 \right)$-complete. 
\end{theorem}

We will then proceed with an intuitive explanation regarding proof of \Cref{thm:RCD-is-StoqMA-complete}. 
\subparagraph{Proof Intuition.} The $\StoqMA$-containment proof is inspired by the SWAP test for distinguishing two quantum states~\cite{BCWdW01}, since it could be thought of as a $\StoqMA$ verification circuit with the maximum acceptance probability $1$. 
We below provide a procedure (see \Cref{fig:RCD-in-StoqMA}) to distinguish two reversible circuits $C_0,C_1$ using a non-negative witness, and such a procedure is apparently a $\StoqMA$ verifier. 
The $\StoqMA$-hardness proof is straightforward: replacing $C_0$ and $C_1$ by identity and $V_x^{\dagger} X_1 V_x$ (see \Cref{fig:modified-StoqMA-verifier}), respectively, where $V_x$ is the given $\StoqMA$ verification circuit. 

\begin{figure}[ht!]
\centering
\begin{minipage}[t]{0.48\textwidth}
\centering
\begin{quantikz}
	\lstick{$\ket{+}$} & \ctrl{1} & \gate{X} & \ctrl{1}\slice{} & \meterD{\ket{+}} \\
	\lstick{$\ket{w}$} & \gate[3,bundle={1,2,3}]{C_0} & \qwbundle[alternate]{} & \gate[3,bundle={1,2,3}]{C_1} & \qwbundle[alternate]{}\\
	\lstick{$\ket{\bar{0}}$} & & \qwbundle[alternate]{} & & \qwbundle[alternate]{} \\
	\lstick{$\ket{\bar{+}}$} & & \qwbundle[alternate]{} & & \qwbundle[alternate]{}
\end{quantikz}
\caption{RCD is in $\StoqMA$}
\label{fig:RCD-in-StoqMA}
\end{minipage}
\begin{minipage}[t]{0.48\textwidth}
\centering
\begin{quantikz}
	\lstick{$\ket{+}$} & \ctrl{1} & \qw\slice{} & \meterD{\ket{+}}  \\
	\lstick{$\ket{w}$} & \gate[3,bundle={1,2,3}]{V_x^{\dagger} X_1 V_x} & \qwbundle[alternate]{} & \qwbundle[alternate]{}  \\
	\lstick{$\ket{\bar{0}}$} & & \qwbundle[alternate]{} & \qwbundle[alternate]{} \\
	\lstick{$\ket{\bar{+}}$} & & \qwbundle[alternate]{} & \qwbundle[alternate]{} 
\end{quantikz}
\caption{RCD is $\StoqMA$-hard}
\label{fig:modified-StoqMA-verifier}
\end{minipage}
\end{figure}

Now we proceed with the technical details. 

\begin{proof}[Proof of \Cref{thm:RCD-is-StoqMA-complete}. ]
	We first show $(\alpha,\beta)$-RCD is $\StoqMA\left(1/2+\alpha/2,1/2+\beta/2 \right)$-hard. 
	Consider a $\StoqMA$ verifier $V_x$ as \Cref{fig:modified-StoqMA-verifier}, let $C_0:=V_x^{\dagger} X_1 V_x$ where the $X$ gate in the middle acts on the output qubit, and let $C_1$ be identity. Then for any witness $\ket{w}$, we obtain: 
	\begin{equation}
	\label{eq:RCD-is-StoqMA-hard}
	\begin{aligned}
		\Pr{V_x \text{ accepts } \ket{w}} &= \bra{w}\bra{\bar{0}}\bra{\bar{+}} \left(V_x^{\dagger} \ketbra{+}{+}_1 V_x\right) \ket{w} \ket{\bar{0}} \ket{\bar{+}};\\
		\innerprod{R_0}{R_1} &= \bra{w}\bra{\bar{0}}\bra{\bar{+}} \left(V_x^{\dagger} X_1 V_x\right) \ket{w} \ket{\bar{0}} \ket{\bar{+}}.
	\end{aligned}
	\end{equation}
	Note that $\ketbra{+}{+} = (X+I)/2$, we thereby complete the $\StoqMA$-hardness proof by \Cref{eq:RCD-is-StoqMA-hard}: 
	$\Pr{V_x \text{ accepts } \ket{w}} = 1/2 + \innerprod{R_0}{R_1}/2$.
	
	Now it is left to show the $\StoqMA\left(1/2+\alpha/2,1/2+\beta/2 \right)$ containment of $(\alpha,\beta)$-RCD. 
	Given reversible circuits $C_0,C_1$, we construct a $\StoqMA$ verifier as \Cref{fig:RCD-in-StoqMA}. 
	Hence, we obtain the state before measuring the output qubit (viz. the red dash line): 
\[{\rm Ctrl-}C_1 \cdot X_1 \cdot {\rm Ctrl-}C_0 \left( \frac{\ket{0}+\ket{1}}{\sqrt{2}} \otimes \ket{w} \ket{\bar{0}} \ket{\bar{+}} \right)
= \frac{1}{\sqrt{2}} \ket{0} \ket{R_0} + \frac{1}{\sqrt{2}} \ket{1} \ket{R_1}
:= \ket{\rm RHS}. \]

We thus complete the $\StoqMA$-containment proof: 
$\Pr{V_x \text{ accepts } \ket{w}} 
= \left\| \ketbra{+}{+}_1 \ket{\rm RHS} \right\|_2^2 
=1/2 + \innerprod{R_0}{R_1}/2. $
\end{proof}

\subsection{Exact Reversible Circuit Distinguishability is $\NP$-complete}
We will prove that the exact variant of the Reversible Circuit Distinguishability is $\NP$-complete. 
Moreover, it will signify that $\StoqMA$ with perfect soundness (even the gap between thresholds $\alpha,1/2$ is arbitrarily small) is in $\NP$. 

\begin{proposition}[Exact RCD is $\NP$-complete]
    \label{prop:exact-rev-circuit-NP}
    Exact Reversible Circuit Distinguishability (RCD), namely $(\alpha,0)$-Reversible Circuit Distinguishability for any $0 \leq \alpha < 1$, is $\NP$-complete. 
\end{proposition}

\begin{proof}[Proof Sketch]
It suffices to show an $\NP$ containment. 
By an analogous idea in~\cite{FGMSZ89}, we could find two matched pairs $(s,r)$ and $(s',r')$ as classical witness, where $s,s'$ are indices of non-zero coordinates in the given witness, and $r,r'$ are random bit strings. 
Specifically, for \textit{yes} instances, there exist two such pairs such that the resulting strings $C_0(s,r)$ \footnote{A reversible circuit takes $(s,r)$ as an input, and permutes it to the other binary string as the output. } and $C_1(s',r')$ are identical; whereas it is evident that no matched pairs exist for \textit{no} instances. 
The details are left in \Cref{subsec:exact-RCD-NP}. 
\end{proof}

As a corollary, \Cref{prop:exact-rev-circuit-NP} will imply $\StoqMA$ with perfect soundness is in $\NP$: 

\begin{corollary}[$\StoqMA$ with perfect soundness is in $\NP$]
	\label{corr:StoqMA-perfect-soundness-is-NP}
	$\mathop{\bigcup_{a>1/2}} \StoqMA\left(a,\tfrac{1}{2}\right) = \NP. $
\end{corollary}

\subparagraph{$\StoqMA$ without any ancillary random bit is in $\NP$. }
In fact, distinguishing reversible circuits without any ancillary random bit is $\NP$-complete. 
By analogous reasoning, we also provide an alternating proof of \textit{Strong Equivalence of Reversible Circuits} is $\coNP$-complete~\cite{Jor14}. 
We leave the detailed proof in \Cref{subsec:StoqMA-wo-randomized-ancilla}. 

\section{Soundness error reduction of $\StoqMA$}
\label{sec:AND-repetition-of-StoqMA}

In this section, we will partially solve \Cref{conj:StoqMA-error-reduction} by providing a procedure that reduces the soundness error of any $\StoqMA$ verifier. 

\begin{theorem}[restated of \Cref{thm-inf:soundness-error-reduction-StoqMA}]
	\label{thm:StoqMA-AND-repetition} For any $r =\poly(n)$, 
	\[\StoqMA\left(\frac{1}{2}+\frac{a}{2},\frac{1}{2}+\frac{b}{2}\right) \subseteq \StoqMA\left( \frac{1}{2}+\frac{a^r}{2}, \frac{1}{2}+\frac{b^r}{2} \right). \]
\end{theorem}

Consequently, \Cref{thm:StoqMA-AND-repetition} infers a direct error reduction for $\StoqMA_1$ by choosing appropriate parameters $a,b,r$.

\begin{corollary}[Error reduction of $\StoqMA_1$]
	\label{corr:StoqMA1-error-reduction}
	For any $s$ such that $1/2 \leq s \leq 1$ and $1-s \geq 1/\poly(n)$, 
	$ \StoqMA(1,s) \subseteq \StoqMA\left( 1,1/2 + 2^{-n} \right).  $
\end{corollary}

\begin{proof}
	Choosing $a,b$ such that $1=1/2+a/2$ and $s = 1/2+b/2$, we have $a=1$ and $b=2s-1$. By \Cref{thm:StoqMA-AND-repetition}, we obtain
	$ \StoqMA\left(\frac{1}{2}+\frac{1}{2}\cdot 1, \frac{1}{2} + \frac{1}{2} (2s-1) \right) \subseteq \StoqMA\left(1, \frac{1}{2} + \frac{1}{2} (2s-1)^r \right). $
	To finish the proof, it remains to choose a parameter $r$ such that $r \geq (n+1)/\log_2\left(1/(2s-1)\right)$, since $(2s-1)^r/2 \leq 2^{-n}$ implies that $2^{-r\log_2\left(1/(2s-1)\right)-1} \leq 2^{-n}$. 
\end{proof}

\subsection{AND-type repetition procedure of a $\StoqMA$ verifier}

\subparagraph{Proof Intuition.} The main idea is doing a parallel repetition of a $\StoqMA$ verifier $V_x$, and taking the conjunction (viz., AND) of the outcomes cleverly.
More concretely, given a $\StoqMA$ verification circuit $V_x$ where $x$ is in $\calL\in\StoqMA$, we result in a new $\StoqMA$ verifier by separately substituting an identity and $V_x^{\dagger} X_1 V_x$ for $C_0$, $C_1$ (as \Cref{fig:modified-StoqMA-verifier}). 
Notice the acceptance probability of a $\StoqMA$ verifier's non-negative witness $\ket{w}$, $\Pr{V_x \text{ accepts } \ket{w}}=\frac{1}{2}+\frac{1}{2}\innerprod{D_0}{D_1}$, is linearly dependent to an inner product between states associated with two distributions $D_0,D_1$ where $\ket{D_0}:=\ket{w}\ket{\bar{0}}\ket{\bar{+}}$ and $\ket{D_1}:=V_x \ket{w}\ket{\bar{0}}\ket{\bar{+}}$.  
We could then take advantage of this new $\StoqMA$ verifier by running $r=\poly(|x|)$ copies of these reversible circuits parallelly with the same target qubit, which is denoted as $V'_x$ (see \Cref{fig:AND-repetition-StoqMA}). 

For \textit{yes} instances, it follows that an inner product of two tensor products of distributions is equal to the product of inner products of states associated with these distributions, namely, $\Pr{V'_x \text{ accepts } \ket{w}}=\frac{1}{2}+\frac{1}{2}\innerprod{D_0}{D_1}^r$. 
However, it seems problematic for \textit{no} instances, since a dishonest prover probably wants to cheat with an entangled witness instead of a tensor product among repetitive verifiers. 
We resolve this issue by an observation used in the $\QMA$ error reduction~\cite{KSV02}: the maximum acceptance probability of a verifier $V_x$ is the same as the maximum eigenvalue of a projection $\Pi_{0} V_x^{\dagger} \Pi_1 V_x \Pi_{0}$ where $\Pi_1$ is the final measurement on the designated output qubit and $\Pi_{0}:=\ketbra{\bar{0}}{\bar{0}}\otimes\ketbra{\bar{+}}{\bar{+}}$. 
Eventually, an entangled witness will not help a dishonest prover. 
This is because the maximum eigenvalue of the tensor product of the projection $\Pi_{0} V_x^{\dagger} \Pi_1 V_x \Pi_0$ is also the product of the maximum eigenvalue of this projection. 

Finally, we proceed with the proof of \Cref{thm:StoqMA-AND-repetition}. 

\begin{figure}[t!]
	\centering
	\begin{quantikz}
		\lstick{$\ket{+}$}  & \ctrl{1} & \ \ldots\ \qw & \ctrl{5} & \meterD{\ket{+}} \\
		\lstick{$\ket{w^{(1)}}$} & \gate[3,bundle={1,2,3}]{V_x^{\dagger} X_1 V_x} &\ \ldots\ \qwbundle[alternate]{} & \qwbundle[alternate]{} & \qwbundle[alternate]{} \\
		\lstick{$\ket{\bar{0}}$} & &\ \ldots\ \qwbundle[alternate]{} & \qwbundle[alternate]{} & \qwbundle[alternate]{} \\
		\lstick{$\ket{\bar{+}}$} & &\ \ldots\ \qwbundle[alternate]{} & \qwbundle[alternate]{} & \qwbundle[alternate]{} \\
		\wave&&&&\\
		\lstick{$\ket{w^{(r)}}$} & \qwbundle[alternate]{} &\ \ldots\ \qwbundle[alternate]{} & \gate[3,bundle={1,2,3}]{V_x^{\dagger} X_1 V_x} & \qwbundle[alternate]{}\\
		\lstick{$\ket{\bar{0}}$} & \qwbundle[alternate]{} &\ \ldots\ \qwbundle[alternate]{} & \qwbundle[alternate]{} & \qwbundle[alternate]{} \\
		\lstick{$\ket{\bar{+}}$} & \qwbundle[alternate]{} &\ \ldots\ \qwbundle[alternate]{} & \qwbundle[alternate]{} & \qwbundle[alternate]{}
    \end{quantikz}
	\caption{AND-type repetition procedure of a $\StoqMA$ verifier}
	\label{fig:AND-repetition-StoqMA}
\end{figure}
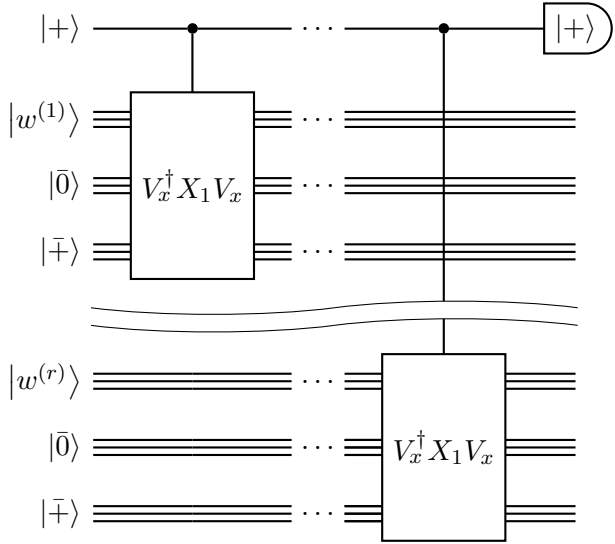

\begin{proof}[Proof of \Cref{thm:StoqMA-AND-repetition}. ]
Given a promise problem $\calL=(\calL_{\rm yes},\calL_{\rm no}) \in \StoqMA(1/2+a/2,1/2+b/2)$. 
For any input $x\in \calL$, we have a $\StoqMA$ verifier $V_x$ which is equivalent to a new $\StoqMA$ verifier $\tilde{V}_x$ as \Cref{fig:modified-StoqMA-verifier}, by the $\StoqMA$-hardness proof of reversible circuit distinguishability as \Cref{thm:RCD-is-StoqMA-complete}. Namely, $\tilde{V}_x$ is starting on a $\ket{+}$ ancillary qubit, applying a controlled-unitary $V_x^{\dagger} X_1 V_x$ on $n_w+n_0+n_+$ qubits, and measuring the designated output qubit. 

Let $\ket{R_w}:=\ket{w}\ket{\bar{0}}\ket{\bar{+}}$ where $\ket{w}$ is a witness, we obtain
\vspace{-0.3em}
\begin{equation}
	\label{eq:PaccX-modified-StoqMA-verifier}
	\left\| \ketbra{+}{+}_1 \left( \frac{1}{\sqrt{2}} \ket{0} \otimes \ket{R_w} + \frac{1}{\sqrt{2}} \ket{1} \otimes \left(V_x^{\dagger} X_1 V_x\right) \ket{R_w} \right) \right\|_2^2 
	= \| \ketbra{+}{+}_1 V_x \ket{R_w}\|_2^2. 
\end{equation}
\vspace{-0.5em}

By an observation used in the $\QMA$ error reduction, namely Lemma 14.1 in~\cite{KSV02}, we notice that the maximum acceptance probability of a $\StoqMA$ verifier $V_x$ is proportion to the maximum eigenvalue of a matrix $M_x:=\bra{\bar{0}} \bra{\bar{+}} V_x^{\dagger} X_1 V_x \ket{\bar{0}} \ket{\bar{+}}$ associated with $V_x$: 
\begin{equation}
\label{eq:StoqMA-verifier-maxeigenval}
\Pr{V_x \text{ accepts } \ket{w}} = \frac{1}{2} + \frac{1}{2} \max_{\ket{w}} \Tr(M_x \ketbra{w}{w}) =  \frac{1}{2} + \frac{1}{2} \lambda_{\max}(M_x). 
\end{equation}

\subparagraph{AND-type repetition procedure of a $\StoqMA$ verifier. }
We now construct a new $\StoqMA$ verifier $V'_x$ using $r$ copies of the witness $\ket{w}$ on $r(n_w+n_0+n_+)+1$ qubits. 
As \Cref{fig:AND-repetition-StoqMA}, $V'_x$ is starting from a $\ket{+}$ ancillary qubit as a control qubit, then applying controlled-unitary $V_x^{\dagger} X_1 V_x$ on qubits associated with different copies of the witness $\ket{w^{(i)}}$ for any $1 \leq i \leq r$.

By an analogous calculation of \Cref{eq:PaccX-modified-StoqMA-verifier}, we have derived the acceptance probability of a witness $w^{(1)} \otimes \cdots \otimes w^{(k)}$ of the new $\StoqMA$ verifier $V'_x$: 
\[\begin{aligned}
    \Pr{V'_x \text{ accepts } \left( w^{(1)}\otimes \cdots \otimes w^{(r)} \right) }
    = \frac{1}{2} + \frac{1}{2} \Tr\left(\ketbra{w^{(i)}}{w^{(i)}} M_x^{\otimes r} \right),
\end{aligned}\]
where $M_x$ is defined in \Cref{eq:StoqMA-verifier-maxeigenval}. 
Hence, the maximum acceptance probability of $V'_x$: 
\begin{equation}
    \label{eq:AND-error-reduction-tensor-product-pacc}
    \max_{\ket{w'}} \Pr{V'_x \text{ accepts } \ket{w'}}
	= \frac{1}{2} + \frac{1}{2} \lambda_{\max} \left( M_x^{\otimes r} \right)
    = \frac{1}{2} + \frac{1}{2} \left( \lambda_{\max}(M_x) \right)^r,
\end{equation}
where the second equality thanks to the property of the tensor product of matrices. 
\Cref{eq:AND-error-reduction-tensor-product-pacc} indicates that entangled-state witnesses are harmless since any entangled-state witness' acceptance probability is not larger than a tensor-product state witness'. 

\vspace{2em}
Finally, we complete the proof by analyzing the maximum acceptance probability of the new $\StoqMA$ verifier $V'_x$ regarding the promises: For \textit{yes} instances, we obtain $\lambda_{\max}(M_x) \geq a$ since there exists $\ket{w}$ such that $\Pr{V_x \text{ accepts } \ket{w}} \geq 1/2+a/2$. By \Cref{eq:AND-error-reduction-tensor-product-pacc}, we have derived 
		$ \Pr{V'_x \text{ accepts } \ket{w}^{\otimes r}} = \frac{1}{2}+\frac{1}{2} \left(\lambda_{\max}(M_x)\right)^r \geq \frac{1}{2}+\frac{a^r}{2}. $ 
For \textit{no} instances, we have $\lambda_{\max}(M_x) \leq b$ since $\Pr{V_x \text{ accepts } \ket{w}} \leq 1/2+b/2$ for all witness $\ket{w}$. By \Cref{eq:AND-error-reduction-tensor-product-pacc}, we further deduce  
		$ \forall w', \Pr{V'_x \text{ accepts } \ket{w'}} =\frac{1}{2}+\frac{1}{2} \left( \lambda_{\max}(M_x) \right)^r \leq \frac{1}{2}+\frac{b^r}{2}. $
\end{proof}

\section*{Acknowledgments}
\noindent
The author thanks Alex B. Grilo for his contribution during the early stage of \Cref{subsec:reversible-circuit-StoqMA}, and the proof of \Cref{prop:cStoqMA-in-MA}. The author also thanks anonymous reviewers for pointing out an error in the proof of \Cref{prop:StoqMA1-in-eStoqMA1} and valuable suggestions. 
Additionally, the author thanks Dorit Aharonov for helpful discussion. 
The author was supported by ISF Grant No. 1721/17 when he was affiliated with the Hebrew University of Jerusalem.
Circuit diagrams were drawn by the Quantikz package \cite{Kay18}.  

\bibliographystyle{alpha}
\bibliography{StoqMA-distribution-testing}

\newcommand{\etalchar}[1]{$^{#1}$}
\begin{thebibliography}{FKYYL{\etalchar{+}}16}

\bibitem[AG21]{AG20}
Dorit Aharonov and Alex~B Grilo.
\newblock Two combinatorial ma-complete problems.
\newblock In {\em 12th Innovations in Theoretical Computer Science Conference
  (ITCS 2021)}. Schloss Dagstuhl-Leibniz-Zentrum f{\"u}r Informatik, 2021.

\bibitem[AGL20]{AGL20}
Dorit Aharonov, Alex~B Grilo, and Yupan Liu.
\newblock $\mathsf{StoqMA}$ vs. $\mathsf{MA}$: the power of error reduction.
\newblock {\em arXiv preprint arXiv:2010.02835}, 2020.

\bibitem[AKKT20]{AKKT19}
Scott Aaronson, Robin Kothari, William Kretschmer, and Justin Thaler.
\newblock Quantum lower bounds for approximate counting via laurent
  polynomials.
\newblock In {\em Proceedings of the 35th Computational Complexity Conference},
  pages 1--47, 2020.

\bibitem[AR20]{AR20}
Scott Aaronson and Patrick Rall.
\newblock Quantum approximate counting, simplified.
\newblock In {\em Symposium on Simplicity in Algorithms}, pages 24--32. SIAM,
  2020.

\bibitem[Bab85]{Bab85}
L{\'a}szl{\'o} Babai.
\newblock Trading group theory for randomness.
\newblock In {\em Proceedings of the seventeenth annual ACM symposium on Theory
  of computing}, pages 421--429, 1985.

\bibitem[BBT06]{BBT06}
Sergey Bravyi, Arvid~J Bessen, and Barbara~M Terhal.
\newblock Merlin-arthur games and stoquastic complexity.
\newblock {\em arXiv preprint quant-ph/0611021}, 2006.

\bibitem[BCWdW01]{BCWdW01}
Harry Buhrman, Richard Cleve, John Watrous, and Ronald de~Wolf.
\newblock Quantum fingerprinting.
\newblock {\em Physical Review Letters}, 87(16):167902, 2001.

\bibitem[BDOT08]{BDOT06}
Sergey Bravyi, David~P Divincenzo, Roberto Oliveira, and Barbara~M Terhal.
\newblock The complexity of stoquastic local hamiltonian problems.
\newblock {\em Quantum Information \& Computation}, 8(5):361--385, 2008.

\bibitem[BDRV19]{BDRV19}
Itay Berman, Akshay Degwekar, Ron~D Rothblum, and Prashant~Nalini Vasudevan.
\newblock Statistical difference beyond the polarizing regime.
\newblock In {\em Theory of Cryptography Conference}, pages 311--332. Springer,
  2019.

\bibitem[BGM06]{BGM06}
Elmar B{\"o}hler, Christian Gla{\ss}er, and Daniel Meister.
\newblock Error-bounded probabilistic computations between $\mathsf{MA}$ and
  $\mathsf{AM}$.
\newblock {\em Journal of Computer and System Sciences}, 72(6):1043--1076,
  2006.

\bibitem[BH17]{BH17}
Sergey Bravyi and Matthew Hastings.
\newblock On complexity of the quantum ising model.
\newblock {\em Communications in Mathematical Physics}, 349(1):1--45, 2017.

\bibitem[BHMT02]{BHMT02}
Gilles Brassard, Peter Hoyer, Michele Mosca, and Alain Tapp.
\newblock Quantum amplitude amplification and estimation.
\newblock {\em Contemporary Mathematics}, 305:53--74, 2002.

\bibitem[Bra15]{Bravyi15}
Sergey Bravyi.
\newblock Monte carlo simulation of stoquastic hamiltonians.
\newblock {\em Quantum Information \& Computation}, 15(13-14):1122--1140, 2015.

\bibitem[BT10]{BT10}
Sergey Bravyi and Barbara Terhal.
\newblock Complexity of stoquastic frustration-free hamiltonians.
\newblock {\em SIAM Journal on Computing}, 39(4):1462--1485, 2010.

\bibitem[Can20]{Can20}
Cl{\'e}ment~L Canonne.
\newblock A survey on distribution testing: Your data is big. but is it blue?
\newblock {\em Theory of Computing}, pages 1--100, 2020.

\bibitem[CGL{\etalchar{+}}20]{CGLLTW19}
Nai-Hui Chia, Andr{\'a}s Gily{\'e}n, Tongyang Li, Han-Hsuan Lin, Ewin Tang, and
  Chunhao Wang.
\newblock Sampling-based sublinear low-rank matrix arithmetic framework for
  dequantizing quantum machine learning.
\newblock In {\em Proceedings of the 52nd Annual ACM SIGACT Symposium on Theory
  of Computing}, pages 387--400, 2020.

\bibitem[CM16]{CM16}
Toby Cubitt and Ashley Montanaro.
\newblock Complexity classification of local hamiltonian problems.
\newblock {\em SIAM Journal on Computing}, 45(2):268--316, 2016.

\bibitem[CR14]{CR14}
Cl{\'e}ment Canonne and Ronitt Rubinfeld.
\newblock Testing probability distributions underlying aggregated data.
\newblock In {\em International Colloquium on Automata, Languages, and
  Programming}, pages 283--295. Springer, 2014.

\bibitem[DKW18]{DKW18}
Constantinos Daskalakis, Gautam Kamath, and John Wright.
\newblock Which distribution distances are sublinearly testable?
\newblock In {\em Proceedings of the Twenty-Ninth Annual ACM-SIAM Symposium on
  Discrete Algorithms}, pages 2747--2764. SIAM, 2018.

\bibitem[FGM{\etalchar{+}}89]{FGMSZ89}
Martin Furer, Oded Goldreich, Yishay Mansour, Michael Sipser, and Stathis
  Zachos.
\newblock On completeness and soundness in interactive proof systems.
\newblock {\em Advainces in Computing Research: A Research Annual,},
  5:429--442, 1989.

\bibitem[FKYYL{\etalchar{+}}16]{FKLMN16}
Bill Fefferman, Hirotada Kobayashi, Cedric Yen-Yu~Lin, Tomoyuki Morimae, and
  Harumichi Nishimura.
\newblock Space-efficient error reduction for unitary quantum computations.
\newblock In {\em 43rd International Colloquium on Automata, Languages, and
  Programming (ICALP 2016)}. Schloss Dagstuhl-Leibniz-Zentrum fuer Informatik,
  2016.

\bibitem[FL18]{FL18}
Bill Fefferman and Cedric Yen-Yu Lin.
\newblock A complete characterization of unitary quantum space.
\newblock In {\em 9th Innovations in Theoretical Computer Science Conference
  (ITCS 2018)}. Schloss Dagstuhl-Leibniz-Zentrum fuer Informatik, 2018.

\bibitem[Gri20]{Grilo20}
Alex~B. Grilo.
\newblock Private communication, 2020.

\bibitem[GS86]{GS86}
Shafi Goldwasser and Michael Sipser.
\newblock Private coins versus public coins in interactive proof systems.
\newblock In {\em Proceedings of the eighteenth annual ACM symposium on Theory
  of computing}, pages 59--68, 1986.

\bibitem[HJ12]{HJ12}
Roger~A Horn and Charles~R Johnson.
\newblock {\em Matrix analysis}.
\newblock Cambridge university press, 2012.

\bibitem[IKW02]{IKW02}
Russell Impagliazzo, Valentine Kabanets, and Avi Wigderson.
\newblock In search of an easy witness: Exponential time vs. probabilistic
  polynomial time.
\newblock {\em Journal of Computer and System Sciences}, 65(4):672--694, 2002.

\bibitem[Jor14]{Jor14}
Stephen~P Jordan.
\newblock Strong equivalence of reversible circuits is
  $\mathsf{coNP}$-complete.
\newblock {\em Quantum Information \& Computation}, 14(15-16):1302--1307, 2014.

\bibitem[JWB05]{JWB05}
Dominik Janzing, Pawel Wocjan, and Thomas Beth.
\newblock "non-identity-check" is $\mathsf{QMA}$-complete.
\newblock {\em International Journal of Quantum Information}, 3(03):463--473,
  2005.

\bibitem[Kay18]{Kay18}
Alastair Kay.
\newblock Tutorial on the quantikz package.
\newblock {\em arXiv preprint arXiv:1809.03842}, 2018.

\bibitem[Kit99]{Kit99}
Alexei Kitaev.
\newblock Quantum $\mathsf{NP}$.
\newblock {\em Talk at AQIP}, 99, 1999.

\bibitem[KSV02]{KSV02}
Alexei~Yu Kitaev, Alexander Shen, and Mikhail~N Vyalyi.
\newblock {\em Classical and quantum computation}.
\newblock American Mathematical Soc., 2002.

\bibitem[KvM02]{KvM02}
Adam~R Klivans and Dieter van Melkebeek.
\newblock Graph nonisomorphism has subexponential size proofs unless the
  polynomial-time hierarchy collapses.
\newblock {\em SIAM Journal on Computing}, 31(5):1501--1526, 2002.

\bibitem[MN17]{MN17}
Tomoyuki Morimae and Harumichi Nishimura.
\newblock Merlinization of complexity classes above bqp.
\newblock {\em Quantum Information \& Computation}, 17(11-12):959--972, 2017.

\bibitem[MV05]{MV05}
Peter~Bro Miltersen and N~Variyam Vinodchandran.
\newblock Derandomizing arthur--merlin games using hitting sets.
\newblock {\em Computational Complexity}, 14(3):256--279, 2005.

\bibitem[NC02]{NC02}
Michael~A Nielsen and Isaac Chuang.
\newblock Quantum computation and quantum information, 2002.

\bibitem[RS09]{RS09}
Ronitt Rubinfeld and Rocco~A Servedio.
\newblock Testing monotone high-dimensional distributions.
\newblock {\em Random Structures \& Algorithms}, 34(1):24--44, 2009.

\bibitem[SV03]{SV03}
Amit Sahai and Salil Vadhan.
\newblock A complete problem for statistical zero knowledge.
\newblock {\em Journal of the ACM (JACM)}, 50(2):196--249, 2003.

\bibitem[Tan10]{Tan10}
Yu~Tanaka.
\newblock Exact non-identity check is $\mathsf{NQP}$-complete.
\newblock {\em International Journal of Quantum Information}, 8(05):807--819,
  2010.

\bibitem[Tan19]{Tang19}
Ewin Tang.
\newblock A quantum-inspired classical algorithm for recommendation systems.
\newblock In {\em Proceedings of the 51st Annual ACM SIGACT Symposium on Theory
  of Computing}, pages 217--228, 2019.

\bibitem[VO21]{VO20}
Ramgopal Venkateswaran and Ryan O'Donnell.
\newblock Quantum approximate counting with nonadaptive grover iterations.
\newblock In Markus Bl{\"{a}}ser and Benjamin Monmege, editors, {\em 38th
  International Symposium on Theoretical Aspects of Computer Science, {STACS}
  2021, March 16-19, 2021, Saarbr{\"{u}}cken, Germany (Virtual Conference)},
  volume 187 of {\em LIPIcs}, pages 59:1--59:12. Schloss Dagstuhl -
  Leibniz-Zentrum f{\"{u}}r Informatik, 2021.

\bibitem[Vol20]{Vol20}
Ilya Volkovich.
\newblock The untold story of $\mathsf{SBP}$.
\newblock In {\em International Computer Science Symposium in Russia}, pages
  393--405. Springer, 2020.

\bibitem[Wat00]{Wat00}
John Watrous.
\newblock Succinct quantum proofs for properties of finite groups.
\newblock In {\em Proceedings 41st Annual Symposium on Foundations of Computer
  Science}, pages 537--546. IEEE, 2000.

\bibitem[Wat16]{Wat16}
Thomas Watson.
\newblock The complexity of estimating min-entropy.
\newblock {\em Computational Complexity}, 25(1):153--175, 2016.

\end{thebibliography}

\appendix
\section{Missing proofs}
\subsection{Proof of \Cref{prop:cStoqMA-in-MA}: $\cStoqMA \subseteq \MA$}
\label{subsec:cStoqMA-proofs}
\begin{proof}[Proof of \Cref{prop:cStoqMA-in-MA}. ]
Given a $\cStoqMA$ verifier $V_x$ on $n=n'+n_0+n_w$ qubits where $n'$ is the number of qubits of a witness, we construct a new $\MA$ verifier $\tilde{V}_x$ on $n=n'+n_0+n_w$ qubits: first run the verification circuit $V_x$ (without measuring the output qubit), then apply an $X$ gate on the output qubit, after that run the verification circuit's inverse $V_x^{\dagger}$, finally measure the first $n'+n_0$ qubits in the computational basis; $\tilde{V}_x$ accepts iff the first $n'$ bits of the measurement outcome is exactly $s_1\cdots s_{n'}$ and the remained bits are all zero. 

We then calculate the acceptance probability of a classical witness $\ket{s}$ of a $\cStoqMA$ verifier $V_x$, where $w=w_1\cdots w_{n'} \in \binset^{n'}$. Notice $\ketbra{+}{+}=\frac{1}{2} \left( I+X \right)$, we obtain
\begin{equation}
\label[equation]{eq:cStoqMA-in-MA-LHS}
\begin{aligned}
	\Pr{V_x \text{ accepts } s} &= \| \ketbra{+}{+}_1 V_x \ket{s}\ket{\bar{0}}\ket{\bar{+}} \|_2^2\\
	&=\tfrac{1}{2} + \tfrac{1}{2} \bra{s}\bra{\bar{0}}\bra{\bar{+}} V_x^{\dagger} \left( X\otimes I_{n-1} \right) V_x \ket{s}\ket{\bar{0}}\ket{\bar{+}}. 
\end{aligned}
\end{equation}

By a direct calculation, the acceptance probability of a classical witness $\ket{s}$ of $\tilde{V}_x$: 
\begin{equation}
\label[equation]{eq:cStoqMA-in-MA-RHS}
\Pr{\tilde{V}_x \text{ accepts } s} = \innerprod{R}{R} \text{ where } \ket{R}:=\left( \bra{s}\bra{\bar{0}} \otimes I_{n_+} \right) V_x^{\dagger} \left( X\otimes I_{n-1} \right) V_x \ket{s}\ket{\bar{0}}\ket{\bar{+}}. 
\end{equation}

It is evident that $\ket{R}$ is a subset state and ${\rm supp}(\ket{R}) \subseteq \binset^{n_+}$. 
Together with Equations \eqref{eq:cStoqMA-in-MA-LHS} and \eqref{eq:cStoqMA-in-MA-RHS}, we have completed the proof by noticing  
$\Pr{V_x \text{ accepts } s} = \frac{1}{2}+\frac{1}{2} \innerprod{\bar{+}}{R} 
= \frac{1}{2}+\frac{1}{2} \innerprod{R}{R} 
= \frac{1}{2}+\frac{1}{2} \Pr{\tilde{V}_x \text{ accepts } s}. $
\end{proof}

Could we extend \Cref{prop:cStoqMA-in-MA} from a classical witness to a probabilistic witness $\sum_{s_i} \sqrt{D(i)} \ket{s_i}$ with a polynomial-size support\footnote{Such witnesses are clearly easy witnesses, but not all easy witnesses have polynomial-bounded size support. See the explicit construction in \Cref{subsec:StoqMA1-in-eStoqMA} as an example.}? 
Notice that the crucial equality $\bra{\bar{+}}\ket{R} = \innerprod{R}{R}$ utilized in \Cref{prop:cStoqMA-in-MA} does not hold anymore, we need \textit{an efficient evaluation algorithm} calculating $D(i)$ given an index $i$. 
Moreover, we have to calculate each coordinate's contribution on the acceptance probability \textit{separately}, so the accumulated additive error is still supposed to be inverse-polynomial, which indicates the support size of this probabilistic witness is \textit{negligible} for some polynomial. 

\subsection{Classical witness is not optimal for any $\StoqMA_1$ verifier}
\label{subsec:StoqMA-classical-witness-not-optimal}

\begin{proposition}
\label{proposition:StoqMA-classical-witness-not-optimal}
	Classical witness is not optimal for any $\StoqMA_1$ verifier. 
\end{proposition}

\begin{proof}
	Consider a $\StoqMA_1$ verifier $V_x$ that uses only identity gates, then
	\begin{enumerate}[(1)]
		\item For all classical witness $s_i\in\binset^{n_w}$, $\Pr{V_x \text{ accepts } s_i}=\frac{1}{2}$ since $\innerprod{R_0}{R_1}=0$ where the resulting state before the measurement is $\ket{0}\otimes \ket{R_0} + \ket{1}\otimes \ket{R_1}$.  
		\item For any classical witness $s_i,s_j\in\binset^{n_w}$ such that $s_i$ and $s_j$ are identical except for the first bit, one can construct a witness $\ket{s}=\frac{1}{\sqrt{2}}\ket{s_i}+\frac{1}{\sqrt{2}}\ket{s_j}$, $\Pr{V_x \text{ accepts } s}=1$ since $\innerprod{R_0}{R_1}=1$. 
	\end{enumerate}
	We thus conclude that classical witness is not optimal for this $\StoqMA_1$ verifier. 
\end{proof}

\subsection{Proof of \Cref{prop:exact-rev-circuit-NP}: Exact RCD is $\NP$-complete}
\label{subsec:exact-RCD-NP}
\begin{proof}[Proof of \Cref{prop:exact-rev-circuit-NP}]
	Exact RCD is $\NP$-hard, namely $\NP \subseteq \StoqMA\left(1,1/2\right)$, straightforwardly follows from the proof of \Cref{proposition:simulating-Z-meas-by-X}. It suffices to prove that the exact RCD is in $\NP$. 
	By \Cref{thm:RCD-is-StoqMA-complete}, $(2\alpha-1,0)$-RCD is $\StoqMA\left(\alpha,1/2\right)$-complete. Let $\ket{w}$ be an $n_w$-qubit non-negative witness such that $\ket{w} := \sum_{s_i \in \supp{w}} \sqrt{D_w(s_i)} \ket{s_i}$, then
	$\Pr{V_x \text{ accepts } \ket{w}} = \frac{1}{2}+\frac{1}{2}\innerprod{R_0}{R_1}=\frac{1}{2}+\frac{1}{2}\bra{w}\bra{\bar{0}} \bra{\bar{+}} C_0^{\dagger} C_1 \ket{w}\ket{\bar{0}} \ket{\bar{+}}.$
	
	For \textit{yes} instances, note that $\innerprod{R_0}{R_1}=2\alpha-1$ and $\alpha > 1/2$, we have derived
	\begin{equation}
		\label{eq:StoqMA-perfect-soundness-inner-prod}
		\innerprod{R_0}{R_1} = \sum_{s_i,s_j\in\supp{w}} \sum_{r,r'\in\binset^{n_+}} \frac{\sqrt{D_w(s_i)D_w(s_j)}}{2^{n_+}} \bra{s_i}\bra{\bar{0}}\bra{r} C_0^{\dagger} C_1 \ket{s_j}\ket{\bar{0}}\ket{r'} > 0.
	\end{equation}
	Since $\forall s_i,s_j$, $D_w(s_i) D_w(s_j) \geq 0$, there exists $s_i,s_j \in \supp{w}$ and $r,r'\in\binset^{n_+}$ such that
	\begin{equation}
		\label{eq:StoqMA-perfect-soundness-yes}
		\bra{s_i}\bra{\bar{0}}\bra{r} C_0^{\dagger} C_1 \ket{s_j}\ket{\bar{0}}\ket{r'} = 1.
	\end{equation} 
	
	For \textit{no} instances, combining $\innerprod{R_0}{R_1}=0$ and \Cref{eq:StoqMA-perfect-soundness-inner-prod}, it infers
	\begin{equation}
		\label{eq:StoqMA-perfect-soundness-no}
		\forall s_i,s_j\in \supp{w}, \forall r,r'\in\binset^{n_+}, \bra{s_i}\bra{\bar{0}}\bra{r} C_0^{\dagger} C_1 \ket{s_j}\ket{\bar{0}}\ket{r'} = 0.
	\end{equation}
		
	We eventually construct an $\NP$ verifier as follows. 
	The input is the classical description of two reversible circuits $C_0$ and $C_1$, and the witness is two pairs of binary strings $(s_0,r_0)$ and $(s_1,r_1)$. The verifier accepts iff $C_0(s_0,0^{n_0},r_0)$ and $C_1(s_1,0^{n_0},r_1)$ are identical where $C_i (i=0,1)$ takes $(s_i,0^{n_0},r_i)$ as an input and permutes it as the output. 
	Notice these strings $s_0,r_0,s_1,r_1$ exists for \textit{yes} instances owing to \Cref{eq:StoqMA-perfect-soundness-yes}, whereas they do not exist for \textit{no} instances due to \Cref{eq:StoqMA-perfect-soundness-no}, which achieves the proof. 
\end{proof}

\subsection{$\StoqMA$ without any ancillary random bit is in $\NP$}
\label{subsec:StoqMA-wo-randomized-ancilla}

\begin{proposition}
\label{prop:StoqMA-wo-plus-in-NP}
	$\StoqMA$ without any ancillary random bit is $\NP$-complete. 
\end{proposition}

\begin{proof}
	It suffices to show that $\StoqMA$ without any ancillary random bit (viz. ancillary qubits which is initially $\ket{+}$) is in $\NP$. 
	As a straightforward corollary of \Cref{thm:RCD-is-StoqMA-complete}, distinguishing reversible circuits without $\ket{+}$ ancillary qubit is complete for $\StoqMA$ without $\ket{+}$ ancillary qubit, which is essentially $\NP$ according to \Cref{sec:prelim-classes}. 
	
	Consider reversible circuits $C_0$ and $C_1$ act on $n_w+n_0$ qubits where $n_0$ is the number of $\ket{0}$ ancillary qubits, we observe that if $C_0$ and $C_1$ are not distinguishable with respect to any classical witness, then $\exists s\in\binset^{n_w}, \bra{s}\bra{\bar{0}}C_0^{\dagger}C_1 \ket{s}\ket{\bar{0}}=1$ since reversible circuits $C_0$ and $C_1$ are bijections. 	Otherwise, it is evident that $\forall w, \bra{w}\bra{\bar{0}}C_0^{\dagger}C_1 \ket{w}\ket{\bar{0}}=0$ provided $C_0$ and $C_1$ are distinguishable with respect to any witness. 
	It is thus sufficient to only consider classical witnesses for distinguishing $C_0$ and $C_1$, namely, classical witness is optimal. 
	
	Now we provide an $\NP$ verifier. The input is the classical description of two reversible circuits $C_0$ and $C_1$, and the witness is a $n_w$-bit string $s$. The verifier accepts iff $C_0(s,0^{n_0})$ is identical to $C_1(s,0^{n_0})$. Note by inspection, the analysis is completed by above showing classical witness is optimal, which finishes the proof. 
\end{proof}

By analogous reasoning, we provide an alternating proof of \cite{Jor14} with respect to the variant of RCD defined in \Cref{remark:RCD-coStoqMA}. 
\begin{proposition}
	Equivalence	check of reversible circuits without any ancillary random bit is $\coNP$-complete. 
\end{proposition}

\begin{proof}
	Consider reversible circuits $C_0,C_1$ act on $n_w+n_0$ qubits, we observe that if $C_0$ and $C_1$ are not exactly equivalent, then $\exists s\in\binset^{n_w}, \bra{s}\bra{\bar{0}}C_0^{\dagger}C_1 \ket{s}\ket{\bar{0}}=0$
	since reversible circuits $C_0$ and $C_1$ are essentially bijections. 
	Otherwise, it is evident that $\forall w, \bra{w}\bra{\bar{0}}C_0^{\dagger}C_1 \ket{w}\ket{\bar{0}}=1$ provided $C_0$ and $C_1$ are exactly equivalent. 
	Therefore, classical witness is optimal, and the remained proof follows from the proof of \Cref{prop:StoqMA-wo-plus-in-NP}. 
\end{proof}

\section{$\SetCSP_{0,1/\poly}$ is $\StoqMA_1$-complete}
\label{sec:SetCSP}

We start from the definition of $\SetCSP$ with frustration: 

\begin{definition}[$k\text{-}\SetCSP_{\epsilon_1,\epsilon_2}$, adapted from Section 4.1 in  \cite{AG20}]
\label{def:SetCSP}
	Given a sequence of $k$-local set-constraints $C=(C_1,\cdots,C_m)$ on $\binset^n$, where $k$ is a constant, $n$ is the number of variables, and $m$ is a polynomial of $n$. A set-constraint $C_i$ acts on $k$ distinct elements of $[n]$, and it consists of a collection $Y(C_i) = \{Y_1^{(i)},\cdots,Y_{l_i}^{(i)}\}$of disjoint subsets $Y_j^{(i)} \subseteq \binset^k$. Promise that one of the following holds, decide whether
	\begin{itemize}
		\item {\bf Yes}: There exists a subset $S \subseteq \binset^n$ s.t. $\setunsat(C,S) \leq \epsilon_1(n)$;
		\item {\bf No}: For any subset $S \subseteq \binset^n$, $\setunsat(C,S) \geq \epsilon_2(n)$, 
	\end{itemize}
	where $\epsilon_1$ and $\epsilon_2$ are efficiently computable function and $\epsilon_2-\epsilon_1 \geq 1/\poly(n)$. 
\end{definition}

Now we briefly define a $\SetCSP$ instance $C$'s frustration. We leave the formal definition in \Cref{proposition:SetCSP-frustration}. 
The frustration of a set-constraint $C$ regarding a subset $S$ is 
$\setunsat(C,S) = \frac{1}{m} \sum_{i=1}^m \setunsat(C_i,S) = \frac{1}{m} \sum_{i=1}^m \left(\frac{|B_i(S)|}{|S|} + \frac{|L_{i}(S)|}{|S|}\right), $
where $B_{i}(S)$ is the set of bad strings of $C_i$, namely $\forall s \in B_{i}(S)$, $s|_{\supp{C_i}} \notin \cup_{j=1}^{l_i} Y_j^{(i)}$; And $L_i(S)$ is the set of longing strings of the subset $S$ regarding $C_i$. 

We will prove \Cref{SetCSP-wo-frustration-is-StoqMA1-complete} in the remainder of this section. 

\begin{theorem}
	\label{SetCSP-wo-frustration-is-StoqMA1-complete}
	$\SetCSP_{\negl,1/\poly}$ is $\StoqMA_{1-\negl}$-complete. 
\end{theorem}

\subsection{$\SetCSP_{\negl,1/\poly}$ is $\StoqMA(1-\negl,1/\poly)$-hard}

To prove \Cref{SetCSP-wo-frustration-is-StoqMA1-complete}, we will first show that $\SetCSP_{0,1/\poly}$ is $\StoqMA_1$-hard. 
\begin{proposition}[$\SetCSP$ is hard for $\StoqMA(1-\negl,1/\poly)$]
\label{proposition:SetCSP-is-StoqMA-negl-hard}
	For any super-polynomial $q(n)$ and polynomial $q_1(n)$, there exists a polynomial $q_2(n)$ such that 
	$\SetCSP_{1/q(n),1/p_2(n)}$ is hard for $\StoqMA\left(1-1/q(n),1/p_1(n)\right)$.  
\end{proposition}

\begin{proof}
	The $\StoqMA\left(1-1/q(n),1/p_1(n)\right)$-hardness proof is straightforwardly analogous to the circuit-to-Hamiltonian construction used in $\MA$-hardness proof of $\SetCSP$ in~\cite{AG20}. 
	The only difference is replacing $Y(C^{\rm out})=\{\{00\},\{01\},\{11\}\}$ by
	$Y(C^{\rm out})=\{\{00\},\{01\},\{10,11\}\}$ in Section 4.4.2, since the final measurement on the $(T+1)$-qubit is on the Hadamard basis instead of the computational basis. 
	The rest of the proof follows from an inspection of Section 4.4 in~\cite{AG20}. 
\end{proof}

Then \Cref{corr:SetCSP-wo-frustration-StoqMA1-hard} is an immediate corollary of \Cref{proposition:SetCSP-is-StoqMA-negl-hard} by substituting $0$ for $1/q(n)$: 
\begin{corollary}
	\label{corr:SetCSP-wo-frustration-StoqMA1-hard}
	$\SetCSP_{0,1/\poly}$ is $\StoqMA_1$-hard. 
\end{corollary}

\subsection{$\SetCSP_{a,b}$ is in $\StoqMA(1-a/2,1-b/2)$}
It now remains to show a $\StoqMA_1$ containment of $\SetCSP_{0,1/\poly}$. 
We will complete the proof by mimicking the $\StoqMA$ containment of the stoquastic local Hamiltonian problem in Section 4 in~\cite{BBT06}. 
The starting point is an alternating characterization of the frustration of a set-constraint $C_i$ in a $\SetCSP$ instance $C$. 
The proof of \Cref{proposition:SetCSP-frustration} is deferred in the end of this section. 

\begin{proposition}[Local matrix associated with set-constraint]
\label{proposition:SetCSP-frustration}
	For any $k$-local set-constraint $C_i (1 \leq i \leq m)$, given a subset $S \subseteq \binset^n$, the frustration 
	\[\setunsat(C_i,S) = 1- \sum_{j=1}^{|Y(C_i)|} \sum_{x,y \in Y_j^{(i)}} \tfrac{1}{|Y_j^{(i)}|} \braket{S}{ \left(\ketbra{x}{y} \otimes I_{n-k}\right) }{S}. \]
\end{proposition}

Now we state the $\StoqMA$ containment of $\SetCSP$, as \Cref{lemma:SetCSP-in-StoqMA}. 
\begin{lemma}
\label{lemma:SetCSP-in-StoqMA}
	For any $0 \leq a < b \leq 1$, $\SetCSP_{a,b} \in \StoqMA\left(1-a/2,1-b/2\right)$. 
	Moreover, for a subset $S \subseteq \binset^n$ such that $S = \argmin_{S'}\setunsat(C,S')$, the subset state $\ket{S}$ is an optimal witness of the resulting $\StoqMA$ verifier. 
\end{lemma}

The proof of \Cref{lemma:SetCSP-in-StoqMA} tightly follows from Section 4 in~\cite{BBT06}.
We here provide a somewhat simplified proof using the $\SetCSP$ language by avoiding unnecessary normalization. 

\begin{proof}[Proof of \Cref{lemma:SetCSP-in-StoqMA}. ]
	Given a $\SetCSP_{a,b}$ instance $C=(C_1,\cdots,C_m)$. For each set-constraint $C_i (1 \leq i \leq m)$, we first construct a local Hermitian matrix $M_i$ preserves the frustration, then construct a family of $\StoqMA$ verifiers for such a $M_i$. 
	For any set-constraint $C_i$, we obtain a $k$-local matrix $M_i$ by \Cref{proposition:SetCSP-frustration} such that for any subset $S\subseteq \binset^n$: 
	\begin{equation}
		\label{eq:SetCSP-in-eStoqMA-frustration}
		\setunsat(C_i,S)=1-\braket{S}{M_i \otimes I_{n-k}}{S} \text{ where } M_i=\sum_{j=1}^{|Y(C_i)|} \sum_{x,y \in Y_j^{(i)}} \frac{1}{|Y_j^{(i)}|} \ketbra{x}{y}. 
	\end{equation}
Moreover, for a set $Y_j^{(i)}$ of strings associated with the set-constraint $C_i$, we further have
	\begin{equation}
	\label{eq:SetCSP-in-eStoqMA-decomposition}
	\begin{aligned}
		\sum_{x,y \in Y_j^{(i)}} \ketbra{x}{y} &= \sum_{x\in Y_j^{(i)}} \ketbra{x}{x} + \frac{1}{2} \sum_{x\neq y\in Y_j^{(i)}} \left( \ketbra{x}{y} + \ketbra{y}{x} \right)\\
		&= \sum_{x\in Y_j^{(i)}} V_x \ketbra{0}{0}^{\otimes k} V_x^{\dagger} + \frac{1}{2} \sum_{x\neq y \in Y_j^{(i)}} V_{x,y} \left( X\otimes \ketbra{0}{0}^{\otimes k-1} \right) V_{x,y}^{\dagger},
	\end{aligned}
	\end{equation}
where $V_x$ is a depth-$1$ reversible circuit with $X$ such that $\forall x, \ket{x}=V_x\ket{0^k}$, and $V_{x,y}$ is a $O(k)$-depth reversible circuit with CNOT and X such that $\forall x,y, V_{x,y} \ketbra{0^k}{10^{k-1}} V_{x,y}^{\dagger} = \ketbra{x}{y}$. 

Notice that the resulting local observables in \Cref{eq:SetCSP-in-eStoqMA-decomposition} are either $\ketbra{0}{0}^{\otimes k}$ (i.e. a single-qubit computational-basis measurement) or $X\otimes \ketbra{0}{0}^{\otimes k-1}$ (i.e. a single-qubit Hadamard-basis measurement). 
To construct a $\StoqMA$ verifier, we only allow local observables in form $X\otimes I^{\otimes O(k)}$. 
Namely, we are supposed to simulate a computational-basis measurement by an ancillary qubit and a Hadamard-basis measurement, which is achieved by \Cref{proposition:simulating-Z-meas-by-X}. 
\begin{proposition}[Adapted from Lemma 3 in~\cite{BBT06}]
	\label{proposition:simulating-Z-meas-by-X}
	\begin{enumerate}[(1)]
		\item For any integer $k$, there exists an $O(k)$-depth reversible circuit $W$ using $k$ $\ket{0}$ ancillary qubits and a $\ket{+}$ ancillary qubits s.t.  
		\[\forall \ket{\psi}, \braket{\psi}{\ketbra{0}{0}^{\otimes k}}{\psi} = \bra{\psi}\bra{0}^{\otimes k}\bra{+} W^{\dagger} \left(X \otimes I^{\otimes 2k}\right) W \ket{\psi} \ket{0}^{\otimes k} \ket{+}.\]
		\item For any integer $k$, there exists an $O(k)$-depth circuit $V$ using $k-1$ $\ket{0}$ ancillary qubits s.t. 
		\[\forall \ket{\psi}, \braket{\psi}{X\otimes \ketbra{0}{0}^{\otimes k-1}}{\psi} = \bra{\psi} \bra{0}^{\otimes k-1} W^{\dagger} \left( X\otimes I^{\otimes 2k-2} \right) W \ket{\psi} \ket{0}^{\otimes k-1}.\]
	\end{enumerate}
\end{proposition}
It is worthwhile to mention that the gadgets used in the proof (see Section A.4 in \cite{BBT06}) further provide proof of $\MA \subseteq \StoqMA$ that preserves both completeness and soundness parameters. 

\vspace{1.5em}
	Let $\Idx{C_i}$ be the set of indices, and let $\alpha_{(j,x,y)}$ be the weight of an index $(j,x,y)$, 
	{\small
	\[\begin{aligned}
	\Idx{C_i} &:=\left\{(j,x,y): 1 \leq j \leq |Y_i(C)|, (x,y)\in \binom{Y_j^{(i)}}{2} \sqcup \left\{(x,x): x\in Y_j^{(i)}\right\}\right\};\\
	\alpha_{(j,x,y)} &:= \frac{1}{(1+\mathbb{I}(x\neq y))m|Y_j^{(i)}|}, \text{ where the indicator } \mathbb{I}(x\neq y))=1 \Leftrightarrow x\neq y.
	\end{aligned}\]
	}
	Plugging \Cref{proposition:simulating-Z-meas-by-X} and \Cref{eq:SetCSP-in-eStoqMA-decomposition} into \Cref{eq:SetCSP-in-eStoqMA-frustration}, we have derived 
	\begin{equation}
	\label{eq:set-unset-Ci}
	1-\setunsat(C_i,S) = \sum_{l \in \Idx{C_i}} \alpha_l \braket{S}{ \left(\bra{0}^{\otimes k}\bra{+} U_l^{\dagger} \left(X\otimes I^{\otimes 2k}\right) U_l \ket{0}^{\otimes k} \ket{+}\right) \otimes I_{n-k}}{S}.
	\end{equation}
	
    For a $\SetCSP$ instance $C=(C_1,\cdots,C_m)$, by \Cref{eq:set-unset-Ci}, by substituting  $\ketbra{+}{+}=\frac{1}{2}(X+I)$ into \Cref{eq:set-unset-Ci}, we thus arrive at a conclusion that 
	\begin{equation}
	\label{eq:SetCSP-PaccX}
	\begin{aligned}
		\Pr{V_x \text{ accepts } \ket{S}} 
		&= \frac{1}{m} \sum_{i=1}^m \left( 1-\frac{1}{2} \cdot \setunsat(C_i,S) \right) 
		&= 1 - \frac{1}{2} \cdot \setunsat(C,S).
	\end{aligned}
	\end{equation}
	
	Note that the set of $\StoqMA$ verifiers $V_x$ with the same number of input qubits and witness qubits is linear, namely a convex combination of $l$ $\StoqMA$ verifiers $(V_1,p_1),\cdots,(V_l,p_l)$ can be implemented by additional $\ket{+}$ ancillary qubits and controlled $V_i (1 \leq i \leq l)$. 
	Therefore, by \Cref{eq:SetCSP-PaccX}, we conclude that $\forall a,b$, $\SetCSP_{a,b}$ is in $\StoqMA\left(1-a/2,1-b/2\right)$. 
\end{proof}

Finally, we achieve proof of \Cref{proposition:SetCSP-frustration}: 
\begin{proof}[Proof of \Cref{proposition:SetCSP-frustration}. ]
	Given a $k$-local set-constraint $C_i$, 
	the set of good strings $G_i = \sqcup_{1 \leq j \leq |Y(C_i)|} Y_j^{(i)}$, 
	and the set of bad strings $B_i = \binset^{|J(C_i)|} \setminus G_i$. 
	Also, for any subset $S \binset^n$, the set of bad strings in $S$ is $B_i(S)$. 
	By direction calculation, notice that
	{\small
	\begin{equation}
	\label{eq:SetCSP-contribution}
	\begin{aligned}
	     \frac{|B_i(S)|}{|S|} &= \braket{S}{ \left( \sum_{x\in B_i} \ketbra{x}{x}\otimes I_{n-k} \right) }{S}\\
	     \sum_{j=1}^{|Y(C_i)|} \frac{|L_j^{(i)}(S)|}{|S|} &= \braket{S}{ \left( \sum_{x\in G_i} \ketbra{x}{x}\otimes I_{n-k} \right) }{S} - \sum_{j=1}^{|Y(C_i)|} \sum_{x,y\in Y_j^{(i)}} \frac{1}{|Y_j^{(i)}|} \braket{S}{ \left(\ketbra{x}{y} \otimes I_{n-k}\right) }{S}. 
	\end{aligned}
	\end{equation}
	}
	Plugging \Cref{eq:SetCSP-contribution} and $\binset^{|J(C_i)|} = B_i \sqcup G_i$ into $\setunsat(C_i,S) 
		= \frac{|B_i(S)|}{|S|} + \sum_{j=1}^{|Y(C_i)|} \frac{|L_j^{(i)}(S)|}{|S|}$, we then finish the proof. 

\end{proof}

\end{document}